\pdfoutput=1
\documentclass[a4paper,UKenglish,cleveref,autoref,thm-restate,numberwithinsect]{lipics-v2021}

\title{Matching Patterns with Variables under Hamming Distance}

\titlerunning{Matching Patterns with Variables under Hamming Distance}

\author{Pawe\l{}  Gawrychowski}{University of Wroc\l{}aw, Faculty of Mathematics and Computer Science}{gawry@cs.uni.wroc.pl}{https://orcid.org/0000-0002-6993-5440}{}

\author{Florin Manea }{G\"ottingen University, Computer Science Department and Campus-Institut Data Science, Germany}{florin.manea@cs.informatik.uni-goettingen.de}{https://orcid.org/0000-0001-6094-3324}{}

\author{Stefan Siemer}{G\"ottingen University, Computer Science Department, Germany}{stefan.siemer@cs.uni-goettingen.de}{https://orcid.org/0000-0001-7509-8135}{}

\authorrunning{P. Gawrychowski and F. Manea and S. Siemer}

\Copyright{John Q. Public and Joan R. Public} 



\nolinenumbers 

\EventEditors{John Q. Open and Joan R. Access}
\EventNoEds{2}
\EventLongTitle{42nd Conference on Very Important Topics (CVIT 2016)}
\EventShortTitle{CVIT 2016}
\EventAcronym{CVIT}
\EventYear{2016}
\EventDate{December 24--27, 2016}
\EventLocation{Little Whinging, United Kingdom}
\EventLogo{}
\SeriesVolume{42}
\ArticleNo{23}

\usepackage[utf8]{inputenc}
\usepackage{amsthm}
\usepackage{tabularx}
\usepackage{environ}
\usepackage{amssymb}
\usepackage{amsmath}
\usepackage{tikz}
\usetikzlibrary{positioning}
\usepgflibrary{arrows}
\usetikzlibrary{arrows,decorations.pathreplacing}

\makeatletter
\newcommand{\problemtitle}[1]{\gdef\@problemtitle{#1}}
\newcommand{\probleminput}[1]{\gdef\@probleminput{#1}}
\newcommand{\problemquestion}[1]{\gdef\@problemquestion{#1}}
\NewEnviron{problem}{
  \problemtitle{}\probleminput{}\problemquestion{}
  \BODY
  \par\addvspace{.5\baselineskip}
  \noindent
  \begin{tabularx}{\textwidth}{@{\hspace{\parindent}} l X c}
    \hline
    \multicolumn{2}{@{\hspace{\parindent}}l}{\@problemtitle} \\
    \textbf{Input:} & \@probleminput \\
    \textbf{Question:} & \@problemquestion\\
    \hline
  \end{tabularx}
  \par\addvspace{.5\baselineskip}
}
\makeatother

\newcommand{\sig}[1]{\Sigma^{#1}}
\newcommand{\plus}{\mathtt{+}}
\newcommand{\interval}[3]{#1[#2:#3]}
\newcommand{\pos}[2]{#1[#2]}
\newcommand{\len}[1]{|#1|}

\newcommand{\set}[1]{\lbrace #1 \rbrace}
\newcommand{\hdist}[2]{d_{\mathtt{HAM}}(#1,#2)}

\newcommand\triphash{\scalebox{0.8}{\raisebox{0.4ex}{\#\#\#}}}
\newcommand\singlehash{\scalebox{0.8}{\raisebox{0.4ex}{\#}}}

\newcommand{\match}{\mathtt{Match}}
\newcommand{\misMatch}{\mathtt{MisMatch}}
\newcommand{\minMisMatch}{\mathtt{MinMisMatch}}
\newcommand{\alp}{\mathtt{alph}}
\newcommand{\var}{\mathtt{var}}

\newcommand{\regPat}{\mathtt{Reg}}
\newcommand{\oneRepPat}{\mathtt{1RepVar}}
\newcommand{\kRepPat}{\mathtt{kRepVar}}
\newcommand{\oneVarPat}{\mathtt{1Var}}
\newcommand{\nonCrossPat}{\mathtt{NonCross}}
\newcommand{\OV}{\mathtt{OV}}
\newcommand{\CP}{\mathtt{CP}}

\newcommand{\scd}{\mathtt{scd}}
\newcommand{\kScdPat}{\mathtt{kSCD}}
\newcommand{\oneScdPat}{\mathtt{1SCD}}
\newcommand{\kLocPat}{\mathtt{kLOC}}

\hideLIPIcs

\begin{document}
\colorlet{lipicsYellow}{white}

\maketitle

\begin{abstract}
A pattern $\alpha$ is a string of variables and terminal letters. We say that $\alpha$ matches a word $w$, consisting only of terminal letters, if $w$ can be obtained by replacing the variables of $\alpha$ by terminal words. The matching problem, i.e., deciding whether a given pattern matches a given word, was heavily investigated: it is NP-complete in general, but can be solved efficiently for classes of patterns with restricted structure. In this paper, we approach this problem in a generalized setting, by considering approximate pattern matching under Hamming distance. More precisely, we are interested in what is the minimum Hamming distance between $w$ and any word $u$ obtained by replacing the variables of $\alpha$ by terminal words. Firstly, we address the class of regular patterns (in which no variable occurs twice) and propose efficient algorithms for this problem, as well as matching conditional lower bounds. We show that the problem can still be solved efficiently if we allow repeated variables, but restrict the way the different variables can be interleaved according to a locality parameter. However, as soon as we allow a variable to occur more than once and its occurrences can be interleaved arbitrarily with those of other variables, even if none of them occurs more than once, the problem becomes intractable.\looseness=-1
\end{abstract}

\section{Introduction}
A \emph{pattern} (with variables) is a string which consists of \emph{terminal letters} (e.\,g., $\mathtt{a,b,c}$), treated as constants, and \emph{variables} (e.\,g., $x_1, x_2$). A pattern is mapped to a word by substituting the variables by strings of terminals. For example, $x_1 x_1 \mathtt{bab} x_2 x_2$ can be mapped to $\mathtt{aa aa bab bb}$ by the substitution $(x_1 \to \mathtt{aa}, x_2 \to \mathtt{b})$. If a pattern $\alpha$ can  be mapped to a string of terminals $w$, we say that $\alpha$ matches $w$. The problem of deciding whether there exists a substitution which maps a given pattern $\alpha$ to a given word $w$ is called the {\em matching problem.} 

Patterns with variables and their matching problem appear in various areas of theoretical computer science. In particular, the matching problem is a particular case of the satisfiability problem for word equations. These are equations whose both sides are patterns with variables and whose solutions are substitutions that map both sides to the same word \cite{Loth97}; in the pattern matching problem, one side of the input equation is a string of terminals. Patterns with variables occur also in combinatorics on words (e.g., unavoidable patterns~\cite{Loth02}), stringology (e.g., generalized function matching~\cite{ami:gen, ord:apa}), language theory (e.g., pattern languages~\cite{DBLP:journals/jcss/Angluin80}), or database theory (e.g., 
document spanners~\cite{FreydenbergerHolldack2018, Freydenberger2019, FaginEtAl2015, SchmidSchweikardt2021}). In a more practical setting, patterns with variables are used in connection to extended regular expressions with backreferences~\cite{cam:afo, fri:mas, Fre2013, FreydenbergerSchmid2019}, used in various programming languages. \looseness=-1

The \emph{matching problem} is NP-complete \cite{DBLP:journals/jcss/Angluin80} in general. This is especially unfortunate for some computational tasks on patterns which implicitly solve the matching problem and are thus intractable as well. For instance, in algorithmic learning theory, this is the case for the task of computing {\em descriptive patterns} for finite sets of words~\cite{DBLP:journals/jcss/Angluin80,DBLP:journals/tcs/FernauMMS18}. Such descriptive patterns are  useful for the inductive inference of pattern languages, a prominent example of a language class which can be inferred from positive data (see, the survey \cite{shi:pat} and the references therein). This and many other applications of pattern matching provide a good motivation to identify cases in which the matching problem becomes tractable. A natural approach to this task is to consider restricted classes of patterns. A thorough analysis~\cite{ReidenbachS14, shi:pol2,FerSch2015, DBLP:journals/mst/FernauSV16,DBLP:journals/toct/FernauMMS20,schmid13} of the complexity of the matching problem has provided several subclasses of patterns for which the matching problem is in P, when some structural parameters of patterns are bounded by constants. Prominent examples in this direction are patterns with a bounded number of repeated variables occurring in a pattern, patterns with bounded scope coincidence degree~\cite{ReidenbachS14}, or patterns with bounded locality~\cite{DayFMN17}. The formal definitions of these parameters are given in Section \ref{sec:repeated}, and corresponding efficient matching algorithms be found in \cite{DBLP:journals/toct/FernauMMS20,DayFMN17}, but, to give an intuition, we mention that they are all numerical parameters which describe the structure of patterns and parameterize the complexity of the matching algorithms. That is, in all cases, if the respective parameter equals $k$, the matching algorithm runs in $O(n^{ck})$ for some constant $c$, and, moreover, the matching problem can be shown to be $W[1]$-hard w.r.t. the respective parameter. A more general approach \cite{ReidenbachS14} introduces the notion of treewidth of patterns, and shows that the matching problem can be solved in $O(n^{2k+4})$ time for patterns with bounded treewidth $k$. The algorithms resulting from this general theory are less efficient than the specialized ones, while the matching problem remains $W[1]$-hard w.r.t. treewidth of patterns. See also the survey~\cite{DBLP:conf/cwords/ManeaS19}. \looseness=-1

In this paper, we extend the study of patterns which can be matched efficiently to the case of approximate matching: we allow mismatches between the word $w$ and the image of $\alpha$ under a substitution. More precisely, we consider two problems. In the decision problem $\misMatch_P$ we are interested in deciding, for a given pattern $\alpha$ from a class $P$, a given word $w$, and a non-negative integer $\Delta$ whether there exists a variable-substitution $h$ such that the word $h(\alpha)$ has at most $\Delta$ mismatches to the word $w$; in other words, the Hamming distance $\hdist{h(\alpha)}{w}$ between $h(\alpha)$ and $w$ is at most $\Delta$. Alternatively, we consider the corresponding minimisation problem $\minMisMatch_P$ of computing $\hdist{\alpha}{w}=\min\{\hdist{h(\alpha)}{w}\mid h\mbox{ is a substitution of the variables in }\alpha\}$.

As most real-world textual data (e.g., involving genetic data or text written by humans) contains errors, considering string-processing algorithms in an approximate setting is natural and has been heavily investigated. See, e.g., the recent papers \cite{DBLP:conf/focs/Charalampopoulos20,DBLP:conf/icalp/GawrychowskiU18,GawrychowskiU17,Uznanski20}, and the references therein, as well as classical results such as  \cite{DBLP:journals/jal/AmirLP04,Myers1989,DBLP:conf/focs/LandauV85}. Closer to the topic of this paper, the problem of approximate pattern matching was also considered in the context of regular expression matching -- see \cite{DBLP:journals/tcs/BilleF08,Myers1989} and the references therein. Continuing this line of research, we initiate a study of approximate matching problems for patterns with variables. Intuitively, in our problems, we ask if the input word $w$ is a few mismatches away from matching the pattern $\alpha$, i.e., if $w$ can be seen as a slightly erroneous version of a word which exactly matches $\alpha$. \looseness=-1

\textbf{Our Contribution.} Our results are summarized in Table \ref{results}. 
In this table we describe the results we obtained for the problems $\misMatch_P$ and $\minMisMatch_P$ (introduced informally above and formally in Section \ref{sec:Prel}) for a series of classes $P$ of patterns for which the matching problem $\match$ can be solved in polynomial time. The classes $P$ we consider are the following. The class $\regPat$ of regular patterns, which do not contain more than one occurrence of any variable. The class $\oneVarPat$ of unary patterns, which contain several occurrences of a single variable and terminals. The class $\nonCrossPat$ of non-cross-patterns, which can be factorized in multiple $\oneVarPat$-patterns whose variables are pairwise different. The class $\oneRepPat$ of one-repeated-variables, where only one variable (say $x$) is allowed to occur more than once. The classes $\kLocPat$ of $k$-local patterns and $\kScdPat$ of patterns with scope coincidence degree at most $k$, defined formally in Section \ref{sec:repeated}. The class $\kRepPat$ of $k$-repeated-variables, where only $k$ variables are allowed to occur more than once. We also (indirectly) obtain a lower bound for the complexity of $\misMatch$ and $\minMisMatch$ in the case of patterns with treewidth at most $k$. 

Interestingly, for $\regPat$ we obtain matching upper and conditional lower bounds. As regular patterns are, in fact, a particular case of regular expressions, it is worth mentioning that, due to the conditional lower bounds from \cite{DBLP:conf/focs/BackursI16} on exact regular expression matching, it is not to be expected that the general case of matching regular-expressions under Hamming distance can be solved as efficiently as the case of regular patterns. 
Regarding patterns with repeated variables, we note that while in the case when the number of repeated variables, the scope coincidence degree, or the treewidth was bounded by a constant, polynomial-time algorithms for the exact matching problem were obtained. This does not hold in our approximate setting, unless P=NP. Only the locality measure has the same behaviour as in the case of exact matching: $\misMatch_{\kLocPat}$ and $\minMisMatch_{\kLocPat}$ can still be solved in polynomial time for constant $k$. In the simpler case of $\oneRepPat$-patterns, the locality corresponds to the number of $x$-blocks, so, if this is bounded by a constant, the two problems we consider can be solved in polynomial time. 
\noindent 
\begin{table}
{\footnotesize
\caption{Our results are listed in columns 3 and 4. We assume $|w|=n$, $|\alpha|=m$, $|\var(\alpha)|=p.$}
\vspace{-8pt}
\begin{tabular}{|l|l|l|l|}
\hline
\!\!Class & \!\!$\match(w,\alpha)$ & \!\!$\misMatch(w,\alpha,\Delta)$ & \!\!$\minMisMatch (w,\alpha)$ \\ 
\hline
\!\!$\regPat$ &\!\!$O(n)$ [folklore] & \!\!$O(n\Delta)$ & \!\!$O(n\hdist{\alpha}{w})$ \\
 & & \!\!matching cond. lower bound & \!\!matching cond. lower bound \\
\hline
\!\!$\oneVarPat$ & \!\!$O(n)$ [folklore] & \!\!$O(n)$ & \!\!$O(n)$  \\
\hline
\!\!$\nonCrossPat$ & \!\!$O(nm\log n)$ \cite{DBLP:journals/toct/FernauMMS20}& \!\!$O(n^3p)$ & \!\!$O(n^3p)$  \\
\hline
\!\!$\oneRepPat$ & \!\!$O(n^2)$ \cite{DBLP:journals/toct/FernauMMS20} & \!\!$O(n^{k+2}m)$ & \!\!$O(n^{k+2}m)$, PTAS \\
\!\!$k$=\# $x$-blocks\!\! & & \!\!W[1]-hard w.r.t. $k$ &\!\!W[1]-hard w.r.t. $k$ \\
 & & &\!\!no EPTAS (if $FPT\neq W[1]$) \\
\hline
\!\!$\kLocPat$ & \!\!$O(mkn^{2k+1})$ \cite{DayFMN17} & \!\!$O(n^{2k+2}m)$ & \!\!$O(n^{2k+2}m)$ \\
 & \!\!W[1]-hard w.r.t. $k$ & \!\!W[1]-hard w.r.t. $k$ &\!\!W[1]-hard w.r.t. $k$ \\
 & & & \!\!no EPTAS (if $FPT\neq W[1]$)\\
\hline
\!\!$\kScdPat$ & \!\!$O(m^2n^{2k})$ \cite{DBLP:journals/toct/FernauMMS20} & \!\!NP-hard for $k\geq 2$ & \!\!NP-hard for $k\geq 2$ \\
 & \!\!W[1]-hard w.r.t. $k$ & & \\
\hline
\!\!$\kRepPat$ & \!\!$O(n^{2k})$ \cite{DBLP:journals/toct/FernauMMS20}  & \!\!NP-hard for $k\geq 1$ & \!\!NP-hard for $k\geq1$\\
 & \!\!W[1]-hard w.r.t. $k$ & & \\
\hline
\!\!$k$-bounded\!\! & \!\!$O(n^{2k+4})$ \cite{ReidenbachS14} & \!\!NP-hard for $k\geq 3$ & \!\!NP-hard for $k \geq 3$ \\
\!\!treewidth\!\!  & \!\!W[1]-hard w.r.t. $k$ & & \\
\hline
\end{tabular}\label{results}
\vspace{-20pt}
}\end{table}

The paper is organized as follows: after some preliminaries, we present in detail the results on $\regPat$-patterns. Then we overview the results on patterns with repeated variables. 

\textbf{Future Work.} While our results paint a detailed image of the complexity of $\misMatch$ and $\minMisMatch$ for some prominent classes of patterns for which the matching problem can be solved efficiently, some continuations of this work can be easily identified. Following~\cite{DBLP:journals/toct/FernauMMS20}, it would be interesting to try to optimise the algorithms for all classes from the table (except $\regPat$, where the upper and conditional lower bounds match). In the case of $\regPat$, it would be interesting to consider the problem for regular patterns with a constant number of variables; already in the case of two variables (also known as approximate string matching under Hamming distance) the known complexity upper and lower bounds do not match anymore \cite{GawrychowskiU17,Uznanski20}.  Another direction is to consider the two problems for other distance functions (e.g., edit distance) instead of the Hamming distance. Finally, it would be interesting if the applications of pattern matching in the area of algorithmic learning theory can be formulated (and still remain interesting) in this approximate setting. \looseness=-1

\section{Preliminaries} \label{sec:Prel}

Let $\Sigma$ be a finite alphabet of {\em terminal letters}. Let $\sig{\star}$ be the set of all words and $\varepsilon$ the empty word. 
The concatenation of $k$ words $w_1, w_2, \ldots , w_k$ is written $\Pi_{i=0}^k w_i$.
The set $\sig{\plus}$ is defined as $\sig{\star} \setminus \set{\varepsilon}$. For $w\in\sig{\star}$ the length of $w$ is defined the number of symbols of $w$, and denoted as $\len{w}$. Further, let $\sig{n} = \set{w\in\sig{\star}~|~\len{w}=n}$ and $\sig{\leq n} = \bigcup_{i=0}^n \sig{i}$. The letter on position $i$ of $w$, for $1\leq i\leq \len{w}$, is denoted by $\pos{w}{i}$. 
For $w \in\sig{\plus}$ and $x,y,z \in \sig{\star}$, the word $y$ is a factor of $w$, if $w = xyz$; moreover, if $x=\varepsilon$ (respectively, $z=\varepsilon$, then $y$ is called a prefix (respectively, suffix) of $w$. Let $\interval{w}{i}{j} = \pos{w}{i} \cdots \pos{w}{j}$ be the factor of $w$ starting on position $i$ and ending on position $j$; if $i>j$ then $w[i:j]=\varepsilon$. By $[i:j]$ we denote the set $\{i,i+1,\ldots,j\}$ and $D[i:j]$ denotes an array $D$ whose positions are indexed by the numbers in $[i:j]$. 

Let $\mathcal{X} = \set{x_1, x_2, x_3. \ldots}$ be a set of {\em variables}. For the set of terminals $\Sigma$ and the set of variables $\mathcal{X}$ with $\Sigma\cap \mathcal{X} = \emptyset$, a pattern $\alpha$ is a word containing both terminals and variables, i.e., an element of $PAT_{\Sigma} = (\mathcal{X}\cup\Sigma)^{\plus}$. The set of all patterns, over all terminal-alphabets, is denoted $PAT = \bigcup_\Sigma PAT_\Sigma$. Given a word or a pattern $\gamma$, for the smallest sets (w.r.t. inclusion) $B\subseteq \Sigma$ and $Y\subseteq \mathcal{X}$ with $\gamma \in (B\cup Y)^\star$, define the set of terminal symbols in $v$, denoted by $\alp(\gamma) = B$, and the set of variables of $\gamma$, denoted by $\var(v)=Y$. For any symbol $t \in \Sigma \cup \mathcal{X}$ and $\alpha \in PAT_\Sigma$, $\len{\alpha}_t$ denotes the number of occurrences of $t$ in $\alpha$. 

A substitution (on the variables of $\alpha$) is a mapping $h: \var(\alpha) \rightarrow \sig{\star}$. For every $x\in \var(\alpha)$, we say that $x$ is substituted by $h(x)$ and $h(\alpha)$ denotes the word obtained by substituting every occurrence of a variable $x$ in $\alpha$ by $h(x)$ and leaving all the terminals unchanged. We say that the pattern $\alpha$ matches a word $w\in\sig{\plus}$, if there exists a substitution $h: \var(\alpha) \rightarrow \sig{\star}$ such that $h(\alpha) = w$. The Matching Problem is defined for any family of patterns $P \subseteq PAT$:

\begin{problem}
  \problemtitle{Exact Matching Problem for $P$: $\match_P$}
  \probleminput{A pattern $\alpha \in P$, with $|\alpha|=m$, a word $w$, with $|w|=n$.}
  \problemquestion{Is there a substitution $h$ with $h(\alpha) = w$?}
\end{problem}

In this paper, we will consider an extension of the Matching Problem, in which we allow mismatches between the image of the pattern under a substitution and the matched word. 

For words $w_1, w_2 \in \sig{\star}$ with $\len{w_1} = \len{w_2}$, the Hamming distance between $w_1$ and $w_2$ is defined as $\hdist{w_1}{w_2}= \len{\set{\pos{w_1}{i} \neq \pos{w_2}{i}~|~ 1\leq i \leq \len{w_1}}}$. The Hamming Distance describes, therefore, the number of mismatches between two words. For a pattern $\alpha$ and a word $w$, we can define the Hamming Distance between $\alpha$ and $w$ as $\hdist{\alpha}{w}= \min \{\hdist{h(\alpha)}{w}\mid h\mbox{ is a substitution of the variables of }\alpha  \}$.
With these definitions we can introduce two new pattern matching problems for families of pattern $P\subseteq PAT$. In the first problem, we allow for a certain distance $\Delta$ between the image $h(\alpha)$ of $\alpha$ under a substitution $h$ and the target word $w$ instead of searching for an exact matching. In the second problem, we are interested in finding the substitution $h$ such that the number of mismatches between $h(\alpha)$ and the target word $w$ is minimal, over all possible choices of $h$.

\begin{problem}
  \problemtitle{Approximate Matching Decision Problem for $P$: $\misMatch_P$}
  \probleminput{A pattern $\alpha \in P$, with $|\alpha|=m$, a word $w$, with $|w|=n$, an integer $\Delta\leq m$.}
  \problemquestion{Is $\hdist{\alpha}{w} \leq \Delta$?}
\end{problem}

\begin{problem}
  \problemtitle{Approximate Matching Minimisation Problem for $P$: $\minMisMatch_P$}
  \probleminput{A pattern $\alpha \in P$, with $|\alpha|=m$, a word $w$, with $|w|=n$.}
  \problemquestion{Compute $\hdist{\alpha}{w}$.}
\end{problem}

When analysing the number of mismatches between $h(\alpha)$ and $w$ we need to argue about the number of mismatches between corresponding factors of $h(\alpha)$ and $w$, i.e., the factors occurring between the same positions $i$ and $j$ in both words. To simplify the presentations, for a substitution $h$ that maps a pattern $\alpha$ to a word of the same length as $w$, we will call the factors $h(\alpha)[i:j]$ and $w[i:j]$ aligned under $h$. We omit $h$ when it is clear from the context. Moreover, saying that we align a factor $\alpha[i:j]$ to a factor $w[i':j']$ with a minimum number of mismatches, we mean that we are looking for a substitution $h$ such that $|h(\alpha)|=|w|$, $h(\alpha[i:j])$ is aligned to $w[i':j']$ under $h$, and the resulting number of mismatches between $h(\alpha[i:j])$ and $w[i':j']$ is minimal w.r.t. all other choices for the substitution $h$.

We make some preliminary remarks. Firstly, in all the problems we consider here, we can assume that the pattern $\alpha$ starts and ends with variables, i.e., $\alpha =x\alpha' y$, with $\alpha'$ pattern and $x$ and $y$ variables. Indeed, if this would not be the case, we could simply reduce the problems by considering them for inputs $\alpha'$ and the word $w'$ obtained by removing from $w$ the prefix and suffix aligned, respectively, to the maximal prefix of $\alpha$ which contains only terminals and the maximal suffix of $\alpha$ which contains only terminals. Clearly, in the case of the exact-matching problem the respective prefixes (suffixes) of $w$ and $\alpha$ must match exactly, while in the case of the approximate-matching problems one needs to account for the mismatches created by these prefixes and suffixes. So, from now on, we will work under the assumption that the patterns we try to align to words start and end with variables. 

Secondly, solving $\match_P$ is equivalent to solving $\misMatch_P$ for $\Delta=0$. Also, in a general framework, $\minMisMatch_P$ can be solved by combining the solution of the decision problem $\misMatch_P$ with a binary search on the value of $\Delta$. Given that the distance between $\alpha$ and $w$ is at most $n=|w|$, one needs to use the solution for $\misMatch_P$ a maximum of $log~n$ times in order to find the exact distance between $\alpha$ and $w$. Sometimes this can be done even more efficiently, as shown in Theorem \ref{thm:minMisMatch}. On the other hand, solving $\minMisMatch_P$ leads directly to a solution for $\misMatch_P$. 

The computational model we use to describe our results is the standard unit-cost RAM with logarithmic word size: for an input of size $n$, each memory word can hold $\log n$ bits. Arithmetic and bitwise operations with numbers in $[1:n]$ are, thus, assumed to take $O(1)$ time. Numbers larger than $n$, with $\ell$ bits, are represented in $O(\ell/\log n)$ memory words, and working with them takes time proportional to the number of memory words on which they are represented. In all the problems, we assume that we are given a word $w$ and a pattern $\alpha$, with $|w|=n$ and $|\alpha|=m\leq n$, over a terminal-alphabet $\Sigma=\{1,2,\ldots,\sigma\}$, with $|\Sigma|=\sigma\leq n$. The variables are chosen from the set $\{x_1,\ldots, x_n\}$ and can be encoded as integers between $n+1$ and $2n$. That is, we assume that the processed words are sequences of integers (called letters or symbols), each fitting in $O(1)$ memory words. This is a common assumption in string algorithms: the input alphabet is said to be {\em an integer alphabet}. For instance, the same assumption was also used for developing efficient algorithms for $\match$ in \cite{DBLP:journals/tcs/FernauMMS18}. For a more detailed general discussion on this model see, e.g.,~\cite{crochemore}.

\section{Matching Regular Patterns with Mismatches}\label{sec:Reg}
A pattern $\alpha$ is {\em regular} if $\alpha=w_0\prod_{i=1}^M (x_iw_i)$, with $w_i \in \sig{\star}$. The class of regular patterns is denoted by $\regPat$. For example, the pattern $\alpha_0=\mathtt{ab} x \mathtt{ab} yz \mathtt{baab} $, with $\var{\alpha}=\{x,y,z\}$ is in $\regPat.$

In this section we consider $\misMatch_{\regPat}$ and $\minMisMatch_{\regPat}$.

As mentioned already, a solution for $\misMatch_{\regPat}$ with distance $\Delta=0$ is a solution to $\match_{\regPat}$. The latter problem can be solved in $\mathcal{O}(n)$ by a greedy approach. As noted in Section \ref{sec:Prel}, we can assume that $w_0=w_M=\varepsilon$, so $\alpha=\prod_{i=1}^{M-1} (x_iw_i) x_M$. Thus, we identify the last occurrence $w[\ell+1:\ell+|w_{M-1}|]$ of $w_{M-1}$ in $w$, assign the string $w[\ell+|w_{M-1}|+1:n]$ to $x_{M}$, and then recursively match the pattern $\alpha=\prod_{i=1}^{M-2} (x_iw_i) x_{M-1}$ to $w[1:\ell]$. 

In the following, we propose a solution for $\minMisMatch_{\regPat}$ which generalizes this approach. Further, we will show a matching lower bound for any algorithm solving $\minMisMatch_{\regPat}$.

\subsection{Efficient solutions for $\misMatch_\regPat$ and $\minMisMatch_\regPat$.}
An equivalent formulation of $\minMisMatch_\regPat$ is to find factors $w[\ell_i+1:\ell_i+|w_i|]$, with $1\leq  i\leq M-1$, such that $\sum_{i=1}^{M-1}\hdist{w_i}{w[\ell_i+1:\ell_i+|w_i|]}$ is minimum and $\ell_{i}+|w_i|+1\leq \ell_{i+1}$, for all $i\in \{1,\ldots, M-2\}$. In other words, we want to find the $M-1$ factors $w[\ell_i+1:\ell_i+|w_i|]$, with $i$ from $1$ to $M-1$, such that these factors occur one after the other without overlapping in $w$, they correspond (in order, from left to right) to the words $w_i$, for $i$ from $1$ to $M-1$, and the total sum of mismatches between $w[\ell_i+1:\ell_i+|w_i|]$ and $w_i$, added up for $i$ from $1$ to $M-1$, is minimal. 

To approach this problem we need the following data-structures-preliminaries.

    Given a word $w$, of length $n$, we can construct in $O(n)$-time {\em longest common suffix}-data structures which allow us to return in $O(1)$-time the value $LCS_w(i,j)=max\{|v|\mid v\mbox{ is a suffix of both } w[1:i]\mbox{ and }w[1:j]\}$. See \cite{DBLP:conf/icalp/KarkkainenS03,DBLP:journals/jacm/KarkkainenSB06} and the references therein. 
    Given a word $w$, of length $n$, and a word $u$, of length $m$, we can construct in $O(n+m)$-time data structures which allow us to return in $O(1)$-time the value $LCS_{w,u}(i,j)=max\{|v|\mid v\mbox{ is a suffix of both } w[1:i]\mbox{ and }u[1:j]\}$. This is achieved by constructing $LCS_w$-data structures for $wu$, as above, and noting that $LCS_{w,u}(i,j)=\min(LCS_w(i,n+j),j)$. 

The following two lemmas are based on the data structures defined above and the technique called kangaroo-jump \cite{DBLP:conf/focs/LandauV85}.

\begin{restatable}{lemma}{kangSimple}\label{lem:kangSimple}
Let $w$ and $u$, with $|w|=|u|=n$, be two words and $\delta$ a non-negative integer. Assume that, in a preprocessing phase, we have constructed $LCS_{w,u}$-data structures. We can compute $\min(\delta+1,\hdist{u}{w})$ using $\delta+1$ $LCS_{w,u}$ queries, so in $O(\delta)$ time. 
\end{restatable}

\begin{proof}
Let $a=b=m$ and $d=0$. While $a>0$ and $d\leq \delta$ execute the following steps. Compute $h=LCS_{w,u}(a,b)$. If $h<b$, then increment $d$ by $1$, set $a\gets a-h-1$ and $b\gets b-h-1$, and start another iteration of the while-loop. If $h=b$, then set $b\gets 0$ and exit the while-loop.

It is not hard to note that before each iteration of the while loop it holds that $d=\hdist{w[a+1:m]}{u[b+1:m]}$. When the while loop is finished, $d=\min(\hdist{w[i-m+1:i]}{u[1:m]},\delta +1)$. In each iteration we first identify the length $h$ of the longest common suffix of $w[1:a]$ and $u[1:b]$. Then, we jump over this suffix, as it causes no mismatches, and have either traversed completely the words $w$ and $u$ (and we do not need to do anything more), or we have reached a mismatch between $w$ and $u$, on position $a-h=b-h$. In the latter case, we count this mismatch, jump over it, and repeat the process (but only if the number of mismatches is still at most $\delta$). So, in other words, we go through the mismatches of $w$ and $u$, from right to left, and jump from one to the next one using $LCS_{w,u}$ queries. If we have more than $\delta$ mismatches, we do not count all of them, but stop as soon as we have met the $(\delta+1)^{th}$ mismatch. Accordingly, the algorithm is correct. Clearly, we only need $\delta+1$ $LCS_{w,u}$-queries and the time complexity of this algorithm is $O(\delta)$, once the $LCS_{w,u}$-data structures are constructed. 
\end{proof}

\begin{restatable}{lemma}{kang}\label{lem:kang}
Given a word $w$, with $|w|=n$, a word $u$, with $|u|=m<n$, and a non-negative integer $\delta$, we can compute in $O(n\delta)$ time the array $D[m:n]$ with $n-m+1$ elements, where $D[i]=\min(\delta+1, \hdist{w[i-m+1:i]}{u})$.
\end{restatable}

\begin{proof}
We first construct, in linear time, the $LCS_{w,u}$-data structures for the input words. Note that the $LCS_{w,u}$-data structure can be directly used as $LCS_{w[i:i+m-1],u}$ data structure, for all $i\leq n-m+1$. 

Then, for each position $i$ of $w$, with $i\leq m$, we use Lemma \ref{lem:kangSimple} to compute, in $O(\delta)$ time the value $d=\min(\hdist{u}{w[i-m+1:i]},\delta +1)$. We then set $D[i]\gets d$. By the correctness of Lemma \ref{lem:kangSimple}, we get the correctness of this algorithm. Clearly, its time complexity is $O(n\delta)$. 
\end{proof}

The following result is the main technical tool of this section.
\begin{restatable}{theorem}{misMatchThm}\label{thm:misMatch}
$\misMatch_{\regPat}$ can be solved in $O(n\Delta)$ time. For an accepted instance $w,\alpha, \Delta$ of $\misMatch_{\regPat}$ we also compute $\hdist{\alpha}{w}$ (which is upper bounded by $\Delta$).
\end{restatable}
\begin{proof}
Assume $\alpha=\prod_{i=1}^{M-1} (x_iw_i) x_M$ and let $\alpha_\ell=\prod_{i=\ell}^{M-1} (x_iw_i) x_M$, for $\ell\in\{1,\ldots,M-1\}$.

A first observation is that the problem can be solved in a standard way by dynamic programming in $O(nm)$ time.

We only give the main idea behind this approach. We can compute the minimum number of mismatches $T[i][j]$ which can be obtained when aligning the suffix of length $i$ of $w$ to the suffix of length $j$ of $\alpha$, for all $i\leq n$ and $j\leq m$. Clearly, $T[i][j]$ can be computed based on the values $T[i+1][j+1]$ and, if $\alpha[j]$ is a variable, $T[i+1][j]$. The full technicalities of this standard approach are easy to obtain so we do not go into further details.  

We present a more efficient approach below.

Our efficient algorithm starts with a preprocessing phase, in which we compute $LCS_{w,u}$-data structures, where $u=\prod_{i=\ell}^{M-1} w_i$. This allows us to retrieve in constant time answers to $LCS_{w,w_i}$-queries, for $1\leq i\leq M-1$.

In the main phase of our algorithm, we compute an $(M-1)\times \Delta$ matrix $Suf[\cdot][\cdot]$, where, for $\ell \leq M-1$ and $d\leq \Delta$, we have $Suf[\ell][d]=g$ if and only if $w[g..n]$ is the shortest suffix of $w$ with $\hdist{\alpha_\ell}{w[g:n]}\leq d$.

Once more, we note that the elements of $Suf[\cdot][\cdot]$ can be computed by a relatively straightforward dynamic programming approach in $O(nM\Delta)$ time. But, the strategy we present here is more efficient than that. 

In our algorithm, we first use Lemma \ref{lem:kang} to compute $Suf[M-1][\cdot]$ in $O(n\Delta)$ time. We simply run the algorithm of that lemma on the input strings $w$ and $w_{M-1}$ and the integer $\Delta$. We obtain an array $D[\cdot]$, where $D[i]=\min(\Delta+1, \hdist{w[i-|w_{M-1}|+1:i]}{w_{M-1}})$. We now go with~$j$ from $|w_{M-1}|$ to $n$ and, if $D[j]\leq \Delta$, we set $Suf[M-1][D[j]]=j-|w_{M-1}|+1$. It is clear that $h=Suf[M-1][d]$ will be the starting position of the shortest suffix $w[h:n]$ of $w$ such that $\hdist{w_{M-1}x_M}{w[h:n]}\leq d$. Thus, $Suf[M-1][\cdot]$ was correctly computed, and the time needed to do so is $O(n\Delta)$. 

Further, we describe how to compute $Suf[\ell][\cdot]$ efficiently, based on $Suf[\ell+1][\cdot]$ (for $\ell$ from $M-2$ down to $1$). We use the following approach. We go through the positions $i$ of $w$ from right to left and maintain a queue $Q$. When $i$ is considered, $Q$ stores all elements $d$ such that $Suf[\ell][d]$ was not computed yet until reaching that position, but $i<Suf[\ell+1][d]$. Accordingly, the fact that $d$ is in $Q$ means that with a suitable alignment of $w_\ell$ ending on position~$i$, we could actually find an alignment with $\leq d$ mismatches of $\alpha_\ell$ with $w[i-|w_\ell|+1:n]$: when $Q$ contains $d,\ldots,d-t$, for some $t\geq 0$, an alignment of $w_\ell$ to $w[i-|w_\ell|+1:i]$ with $\leq t$ mismatches would lead to an alignment of $\alpha_\ell$ with $w[i-|w_\ell|+1:n]$ with $\leq d$ mismatches by extending the alignment of $\alpha_{\ell+1}$ to $w[Suf[\ell+1][d-t]:n]$. The values $d$ present in $Q$ at some point are ordered increasingly (the older values are larger), the array $Suf[\ell+1][\cdot]$ is also monotonically increasing, and, as $Suf[\ell][d]$ cannot be set before $Suf[\ell][d']$, for any $d$ and $d'$ such that $d'<d$, the queue $Q$ is actually an interval of integers $[new:old]$, where $new$ is the newest element of $Q$, and $old$ the oldest one. When we consider position $i$ of the word, if the alignment of $w_\ell$ ending on position $i$ causes $t$ mismatches, then to be able to set a value $Suf[\ell][d]$, with $d\in Q$, we need to have that $Suf[\ell+1][d-t]>i$. As $Suf[\ell+1][d]>Suf[\ell+1][d-t]$ and $d\in Q$, this means that $d-t\in Q$, so the number of mismatches $t$ must be strictly upper bounded by $|Q|$, in order to be useful. Accordingly, when considering position $i$, we compute the number $t\gets\min\{\hdist{w_\ell}{w[i-|w_\ell|+1:i]},|Q|\}$, and if $t<|Q|$ we set $Suf[\ell][d]\gets i-|w_\ell|+1$ for all $d$ such that $d-t\in Q$; we also eliminate all these elements $d$ from the queue. Before considering a new position $i$, we check if $i=Suf[\ell+1][new-1]$, and, if yes, we insert $new-1$ in $Q$ and update $new\gets new-1$. \looseness=-1

This computation of $Suf[\ell][\cdot]$ is implemented in the following algorithm:
\begin{enumerate}
    \item Initialization: We maintain a queue $Q$, which initially contains only the $\Delta$. \\
    Let $new \gets \Delta$ (this is the top element of the queue).
    \item Iteration: For $i=Suf[\ell+1][\Delta]-1$ down to $|w_\ell|$ we execute the steps a, b, and c:
    \begin{enumerate}
        \item Using Lemma \ref{lem:kangSimple} we compute $t\gets \min(\hdist{u}{w[i-|w_\ell|+1:i]},|Q|)$.
        \item If $t< |Q|$, we remove from $Q$ all elements $d$, such that $d-t \geq new$, and set, for each of them, $Suf[\ell][d]\gets i-|w_\ell|+1$.
        \item If $Suf[\ell+1][top-1]=i$ then we insert $top-1$ in $Q$ and $top \gets top-1$. Else, if $Suf[\ell+1][top-1]=0$ then set $i\gets 0$ and exit the loop.
    \end{enumerate}
    \item Filling-in the remaining positions: Set all the positions of $Suf[\ell][\cdot]$ which were not filled during the above while-loop to $0$. 
\end{enumerate}

The matrix $Suf[\cdot][\cdot]$ is computed correctly by the above algorithm, as it can be shown by the following inductive argument. \looseness=-1

To show that $Suf[\ell][\cdot]$ is computed correctly by our algorithm, under the assumption that $Suf[\ell+1][\cdot]$ was correctly computed, we make several observations. 

Firstly, it is clear that $Suf[\ell+1][d]\leq Suf[\ell+1][d+1]$. Secondly, when computed correctly, $Suf[\ell][d]$ should be the rightmost position $g$ of $w$ such that $\hdist{w[g:n]}{w_\ell}=t\leq d$ and $Suf[\ell+1][d-t]\geq g+|w_\ell|$. Clearly, if $Suf[\ell][d+1]\neq 0$, then $Suf[\ell][d]< Suf[\ell][d+1]$.

Regarding the algorithm described in the main part of the paper, it is important to observe that the queue $Q$ is ordered increasingly (i.e., the newer is an element in $Q$, the smaller it is) and the elements of $Q$ form an interval $[new:old]$. 

Now, let us show the correctness of the algorithm. 

Let $d$ be a non-negative integer, $d\leq \Delta.$ Assume that our algorithm sets $Suf[\ell][d]=g$, with $g>0$. 

This means that $d$ was removed from the queue in step 2.b when the for-loop was executed for $i=g+|w_\ell|-1$. The reason for this removal was that $\hdist{w[g:g+|w_\ell|-1]}{w_\ell}=t\leq |Q|-1$. Hence, in this step we have removed exactly those elements $\delta$ such that $new\leq \delta-t$. Accordingly, we also have that $new\leq d-t$ holds. Let $g'=Suf[\ell+1][new]$. We thus have $g'>i=g+|w_\ell|-1$,  $\hdist{\alpha_{\ell+1}}{w[g':n]}\leq new$, and $\hdist{w_\ell x_\ell}{w[g:g'-1]} = t$. Putting this all together, we get that $\hdist{\alpha_\ell}{w[g:n]}\leq new+t\leq d$. 

Now, assume for the sake of a contradiction, that there exists $g''>g$ such that $\hdist{\alpha_\ell}{w[g'':n]}\leq d$, i.e., $w[g:n]$ is not the shortest suffix $s$ of $w$ such that $\hdist{\alpha_\ell}{s}\leq d$. In this case, there exists $d''$ such that $g''+|w_\ell|-1< Suf[\ell+1][d'']$ and $d''+\hdist{w[g'':g''+|w_\ell|-1]}{w_\ell}\leq d$. Because $d$ is in $Q$ when $i=g+|w_\ell|-1$ is reached in the for-loop, then $d$ must also be in $Q$ when $i''=g''+|w_\ell|-1$ is reached in the for-loop, because $i<i''<Suf[\ell+1][d'']\leq Suf[\ell+1][d]$. In fact, as $Suf[\ell+1][d]\geq Suf[\ell+1][d''] > i''$, it follows that $d''$ must also be in $Q$ when $i''$ is reached. Thus, $q \geq d-d''$ and, as we have seen above, $d-d''\geq \hdist{w[g'':g''+|w_\ell|-1]}{w_\ell}$. Moreover, if $new''$ is the element on the top of the queue when $i''$ is reached, we have that $new''\leq d''$. Hence, $new'' + \hdist{w[g'':g''+|w_\ell|-1]}{w_\ell} \leq d''+\hdist{w[g'':g''+|w_\ell|-1]}{w_\ell} \leq d$. Therefore, when $i''$ was reached, all the conditions needed to remove $d$ from $Q$ and set $Suf[\ell][d]\gets g''$ were met. We have reached a contradiction with our assumption that $g''>g$. 

In conclusion, if our algorithm sets $Suf[\ell][d]=g$, with $g>0$, then $w[g:n]$ is the shortest suffix of $w$ such that $\hdist{w[g:n]}{w_\ell}\leq d$. By an analogous argument as the one used above in our proof by contradiction, we can show that if our algorithm sets $Suf[\ell][d]=0$ then there does not exist any suffix $w[g:n]$ of $w$ such that $\hdist{w[g:n]}{w_\ell}\leq d$. 

This means that our algorithm computing $Suf[\cdot][\cdot]$ is correct.

To finalize the proof of the theorem, we note that, after computing the entire matrix $Suf[\cdot][\cdot]$, we can accept the instance $w,
\alpha, \Delta$ of $\misMatch_{\regPat}$ if and only if there exists $d\leq \Delta$ such that $Suf[1][d]\neq 0$. Moreover, $\hdist{\alpha}{w}=\min(\{d\mid Suf[1][d]\neq 0\}\cup\{+\infty\})$. 

In the following we show that this algorithm works in $O(n\Delta)$ time. We will compute the complexity of this algorithm using amortized analysis. Firstly, we observe that the complexity of the algorithm is proportional to the total number of $LCS_{w,w_\ell}$-queries we compute in step 2.a, for each $\ell \leq M$ or, in other words, over all executions of the algorithm. Now, we observe that when position $i$ of $w$ is considered (for a certain $\ell$), we do $|Q|$ many  $LCS_{w,w_\ell}$-queries. So, this means that we do one query per each current element of $Q$ (and none if $|Q|=0$). Thus, the number of queries corresponding to each pair $(\ell,d)$ which appears in $Q$ at some point equals the number of positions considered between the step when it was inserted in $Q$ and the step when it was removed from $Q$. This means $O(Suf[\ell+1][d]-Suf[\ell][d])$ queries corresponding to $(\ell,d)$. Summing this up for a fixed $d$ and $\ell$ from $1$ to $M-2$ we obtain that the overall number of queries corresponding to a fixed $\delta$ is $O(Suf[M-1][d])=O(n)$. Adding this up for all $d\leq \Delta$, we obtain that the number of $LCS$-queries performed in our algorithm is $O(n\Delta)$. So, together with the complexity of the initialization of $Suf[M-1][\cdot]$, the complexity of this algorithm is $O(n\Delta)$. 

This algorithm outperforms the other two algorithms solving $\minMisMatch_{\regPat}$ which we mentioned, and, for $\Delta=0$, it is a reformulation of the greedy algorithm solving~$\match_{\regPat}$.
\end{proof}

Now it is not hard to show the following result.
\begin{restatable}{theorem}{minMisMatchThm}\label{thm:minMisMatch}
$\minMisMatch_{\regPat}$ can be solved in $O(n\Phi)$ time, where $\Phi=\hdist{\alpha}{w}$.
\end{restatable}

\begin{proof}
We use the algorithm of Theorem \ref{thm:misMatch} for $\Delta=2^i$, for increasing values of $i$ starting with $1$ and repeating until the algorithm returns a positive answer and computes $\Phi=\hdist{\alpha}{w}$. The algorithm is clearly correct. Moreover, the value of $i$ which was considered last is such that $2^{i-1}< \Phi\leq 2^{i}$. So $i=\lceil \log_2 \Phi\rceil$, and the total complexity of our algorithm is $O(n\sum_{i=1}^{\lceil \log_2 \Phi\rceil}2^i)=O(n\Phi)$.
\end{proof}

\subsection{Lower Bounds for $\misMatch_{\regPat}$ and $\minMisMatch_{\regPat}$.}

In order to show that $\minMisMatch_{\regPat}$ and $\misMatch_{\regPat}$ cannot be solved by algorithms running polynomially faster than the algorithms from Theorems \ref{thm:misMatch} and \ref{thm:minMisMatch}, we will reduce the Orthogonal Vectors problem $\OV$ \cite{DBLP:conf/stacs/Bringmann19} to $\misMatch_\regPat$. 
The overall structure of our reduction is similar to the one used for establishing hardness of computing
edit distance~\cite{BackursI18,BringmannK15} or LCS~\cite{BringmannK18}, however we needed to construct
gadgets specific to our problem.
We recall the $\OV$ problem.

\begin{problem}
  \problemtitle{Orthogonal Vectors: $\OV$}
  \probleminput{Two sets $U,V$ consisting each of $n$ vectors from $\set{0,1}^d$, where $d\in \omega(\log n)$.}
  \problemquestion{Do vectors $u \in U, v \in V$ exist, such that $u$ and $v$ are orthogonal, i.e., for all $1 \leq k \leq d$, $v[k] u[k] = 0$ holds?}
\end{problem}
In general, for a vector $u=(u[1],\ldots,u[d]) \in \{0,1\}^d$, the bits $u[i]$ are called coordinates. 
It is clear that, for input sets $U$ and $V$ as in the above definition, one can solve $\OV$ trivially
in $\mathcal{O}(n^2d)$ time. The following conditional lower bound is known for $\OV$.
\begin{lemma}[$\OV$-Conjecture]
$\OV$ can not be solved in $\mathcal{O}(n^{2-\epsilon} d^{c})$ for any $\epsilon > 0$ and constant $c$, unless the Strong Exponential Time Hypothesis (SETH) fails.  
\end{lemma}

See \cite{DBLP:conf/stacs/Bringmann19,DBLP:journals/tcs/Williams05} and the references therein for a detailed discussion regarding conditional lower bounds related to OV. In this context, we can show the following result.
\begin{theorem}
$\misMatch_{\regPat}$ can not be solved in $\mathcal{O}(|w|^h \Delta^g)$ time (or in $\mathcal{O}(|w|^h |\alpha|^g)$ time) with $h+g= 2-\epsilon$ for some $\epsilon>0$, unless the $\OV$-Conjecture fails. 
\end{theorem}

\begin{proof}
We reduce $\mathtt{OV}$ to $\minMisMatch_\regPat$. For this, we consider an instance of $\OV$: $U=\{u_1, \ldots, u_n\}$ and $V=\{v_1, \ldots, v_n\}$, with $U,V \subset \set{0,1}^d$. We transform this $\OV$-instance into a $\misMatch_{\regPat}$-instance $(\alpha,w,\Delta)$, where $\Delta=n(d + 1)-1$. More precisely, we ensure that for the respective $\misMatch_\regPat$-instance, there exists a way to replace the variables with strings leading to exactly $n(d + 1)$ mismatches between the image of $\alpha$ and $w$ if and only if no two vectors $u_i$ and $v_j$ are orthogonal. But, if there exists at least one orthogonal pair of vectors $u_i$ and $v_j$, there also exists a way to replace the variables of $\alpha$ such that the resulting string has strictly less than $n(d+1)$ mismatches to $w$. Both $|w|$ and $|\alpha|$ are in $\mathcal{O}(nd)$, and can be built in $O(nd)$ time. The reduction consists of three main steps. First we will present a gadget for encoding the single coordinates of vectors $u_i$ and $v_i$ from $U$ and $V$, respectively. Then we will show another gadget to encode a full vector of each respective set. And, finally, we will show how to assemble these gadgets of the vectors from set $U$ into the word $w$ and from $V$ into $\alpha$.

\textbf{First gadget}.
Let $u_i=(u_{i}[1], u_{i}[2],\ldots,u_{i}[d])  \in U, v_j=(v_{j}[1], v_{j}[2],\ldots,v_{j}[d]) \in V$ and let $k$ be a position of these vectors. We define the following gadgets:\\
\centerline{$    A'(i_k)= \begin{cases}
        \mathtt{001},& \text{ if} ~u_{i}[k] = 0.\\
        \mathtt{100},& \text{ if}~u_{i}[k] = 1.
    \end{cases} \quad
    B'(j_k) =\begin{cases}
        \mathtt{000},\quad \text{if}~v_{j}[k] = 0.\\
        \mathtt{011},\quad \text{if}~v_{j}[k] = 1.
    \end{cases} \quad
$}

Note that, when aligned, the pair of strings $(A'(i_k), B'(j_k))$ produces exactly one mismatch if and only if $ u_{i}[k] \cdot v_{j}[k] = 0$; otherwise it produces three mismatches. So, $A'(i_k)$ and $B'(j_k)$ encode the single coordinates of $u_i$ and $v_j$ respectively. 

Further, we construct a gadget $X' = \mathtt{010}$ that produces always one mismatch if aligned to any of the strings $B'(j_k)$ corresponding to coordinates $v_j[k]$. 
See also Figure \ref{fig:regulargadgets}. 
\begin{figure}[h!]
    \centering
    \begin{tikzpicture}[scale=0.7, transform shape]
        \node[draw, circle, label=left:{$A'(i_k) = 0$}] (A0) at (0,0) {$\mathtt{001}$};
        \node[draw, circle, below=8mm of A0, label=left:{$A'(i_k) = 1$}] (A1) {$\mathtt{100}$};
        \node[draw, circle, right=30mm of A0, label=right:{$B'(j_k) = 0$}] (B0) {$\mathtt{000}$};
        \node[draw, circle, below=8mm of B0, label=right:{$B'(j_k) = 1$}] (B1) {$\mathtt{011}$};
        \node[draw, circle, below right=4mm of B0, label=right:{$X'$}, xshift=10mm] (X) {$\mathtt{010}$};

        \draw[-] (A0) -- node[below]{$1$} (B0);
        \draw[-] (A0) -- node[left, xshift=-5mm, yshift=-1mm]{$1$} (B1);
        \draw[-] (A1) -- node[right, xshift=5mm, yshift=-1mm]{$1$} (B0);
        \draw[-] (A1) -- node[above]{$3$} (B1);
        \draw[-] (B0) -- node[below]{$1$} (X);
        \draw[-] (B1) -- node[above]{$1$} (X);
        
    \end{tikzpicture}
    \caption{Gadgets for the encoding of single coordinates of the vectors. On each edge we wrote the number of mismatches between the strings in the nodes connected by that edge.}
    \label{fig:regulargadgets}
\vspace{-5pt}
\end{figure}
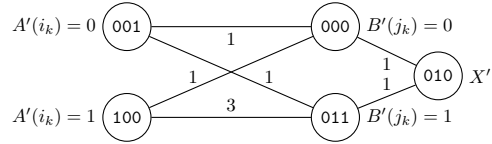

\textbf{Second gadget}.
The gadget $A(i)$ encodes the vector $u_i$, for $1\leq i\leq n$, while the gadget $B(j)$ encodes the vector $v_j$, for $1\leq j\leq n$. We  construct these gadgets such that aligning $B(j)$ to $A(i)$ with a minimum number of mismatches yields exactly $d$ mismatches, if the two corresponding vectors are orthogonal, and exactly $d+1$ mismatches, otherwise. Moreover, we show that any other alignment of the gadgets $B(j)$ with other factors of $w$ yields more mismatches.\looseness=-1

In order to assemble the gadgets $A(i)$ and $B(j)$, for $1\leq i,j\leq n$, we extend the terminal alphabet by three new symbols $\set{\mathtt{a},\mathtt{b},\singlehash}$, as well as use two fresh variables $x_j, y_j$ for each vector $v_j$. The gadgets $A(i)$, for all $i$, and, respectively, the gadgets $B(j)$, for all $j$, consist of the concatenation of the coordinate gadgets $A'(i_k)$ and, respectively, $B'(j_k)$ from left to right, in ascending order of $k$. Each two such consecutive gadgets $A'(i_k)$ and $A'(i_{k+1})$ (respectively, $B'(j_k)$ and $B'(j_{k+1})$) are separated by $\triphash$. We prepend to $A(i)$ the string $\mathtt{bba}$ and append the string $\mathtt{bbb}X$, where $X=(X'\triphash)^{d-1}X'$. In the case of $B(j)$, we prepend $x_j\mathtt{bba}$ and append $y_j$. The full gadgets $A(i)$ and $B(j)$ are defined as follows.
\begin{itemize}
    \item $A(i)=\mathtt{bba}A'(i_1)\triphash A'(i_2)\triphash \ldots A'(i_d)\mathtt{bbb}X$ 
    \item $B(j)= x_j\mathtt{bba}B'(j_1)\triphash B'(j_2)\triphash \ldots B'(j_d)y_j$. 
\end{itemize}
For simplicity of the exposure, let $B'(j)= \mathtt{bba}B'(j_1)\triphash B'(j_2)\triphash \ldots \triphash B'(j_d)$. 

Note that $|A(i)|$ is the same for all $i$, so we can define $M=|A(i)|$.

\textbf{Final assemblage}.
To define the word $w$, we use a new terminal $\mathtt{\$}$. The word $w$ is:
\begin{itemize}
    \item $w =$ $\$^MA(1)\$^{M}A(2)\$^{M}\ldots A(n)\$^{M}A(1)\$^{M}A(2)\ldots \$^{M}A(n)\$^M$
\end{itemize}
To define $\alpha$, we use two new fresh variables $x$ and $y$. The pattern $\alpha$ is:
\begin{itemize}
    \item $\alpha =$ $x\$^{M} B(1)\$^{M}B(2)\$^{M}\ldots \$^{M}B(n)\$^{M}y$.
\end{itemize}

{\bf The correctness of the reduction.} 
We show that there exists a way to align $\alpha$ with $w$ with $<n(d+1)$ mismatches if and only if a pair of orthogonal vectors $u_i\in U$ and $v_j\in V$ exists. Otherwise, there exists an alignment of $\alpha$ to $w$ with exactly $n(d+1)$ mismatches. 

To formally prove that the reduction fulfills this requirement, we proceed as follows.

A general idea: the repetition of the gadgets $A(i)$ in the word $w$ guarantees that, if needed, a pair of gadgets $A(i)$ and $B(j)$, corresponding to the vectors $u_i\in U$ and, respectively, $ v_j \in V$, can be aligned. More precisely, we can align $B'(j)$ to $\mathtt{bba}A'(i_1)\triphash \ldots A'(i_d)$. The variables $x,y$ and $x_j,y_j$, for $j\in\{1,\ldots,n\}$, act as spacers: they allow us to align a string $B'(j)$ to the desired factor of $w$. This kind of alignment is enough for our purposes, as we only need to find one orthogonal pair of vectors, not all of them; however, we need enough space in $w$ for the factors of $\alpha $ occurring before and after $B'(j)$, thus the repetition of the $A(i)$ gadgets. 

We now analyse how a factor $B'(j)$ can be aligned to a factor of $w$. The main idea is to show that if there are no orthogonal vectors, then any alignment of $B'(j)$ to a factor of $w$ creates at least $d+1$ mismatches. Otherwise, we can align it with $d$ mismatches only. 
    
\noindent {\em Case 1:} $B'(j)$ is aligned to a factor $w[i:h]$ of $w$ which starts with $\$ $. Then the prefix $\mathtt{bba}$ of $B'(j)$ causes at least two mismatches, as the first $\mathtt{b}$ in  $\mathtt{bba}$ is aligned to a $\$ $ letter, while the $\mathtt{a}$ is aligned to either a $\mathtt{b}$ letter (from a $\mathtt{bba}$ factor) or a $\$ $ letter. The rest of $B'(j)$ causes, overall, at least $d$ mismatches, one per each group $B'(j_k)$. So, in this case, we have at least $d+2$ mismatches caused by $B'(j)$.

\noindent {\em Case 2:} $B'(j)$ is aligned a factor $w[i:h]$ of $w$ which ends with $\$ $. Then, its prefix $\mathtt{bba}$ cannot be aligned to a factor $\mathtt{bba}$ of $w$. So, the $\mathtt{a}$ of the prefix $\mathtt{bba}$ of $B'(j)$ produces one mismatch, while the suffix $B'(j_d)$ causes at least $2$ mismatches. The rest of $B'(j)$ causes at least $d-1$ mismatches, one per each remaining group $B'(j_k)$. So, in this case, we have again at least $d+2$ mismatches caused by $B'(j)$.
    
\noindent {\em Case 3:} $B'(j)$ is aligned exactly to the factor $\mathtt{bba}A'(i_1)\triphash \ldots A'(i_d)$ and $u_i$ and $v_j$ are orthogonal, then $B'(j)$ causes exactly $d$ mismatches.
    
\noindent {\em Case 4:} $B'(j)$ is aligned exactly to the factor $\mathtt{bba}A'(i_1)\triphash  \ldots A'(i_d)$ and $u_i$ and $v_j$ are not orthogonal, then $B'(j)$ causes at least $d+2$ mismatches.
    
\noindent {\em Case 5:} $B'(j)$ is aligned exactly to the factor $\mathtt{bbb}X$, then $B'(j)$ causes $d+1$ mismatches.
    
\noindent {\em Case 6:} $B'(j)$ is aligned to a factor starting strictly inside $\mathtt{bba}A'(i_1)\triphash \ldots A'(i_d)$, then the prefix $\mathtt{bba}$ of $B'(j)$ cannot be aligned to a factor $\mathtt{bba}$ of $w$, so it causes at least two mismatches (from the alignment of $\mathtt{ba}$). The rest of $B'(j)$ causes at least $d$ mismatches, one per each group $B'(j_k)$. So, overall, $B'(j)$ causes at least $d+2$ mismatches in this case. 

To ease the understanding, cases 3 and 4 are illustrated in the following table: when aligning $A(i)$ to $B(j)$, to obtain the desired number of mismatches, we can match the parts of $A(i)$ to the parts of $B(j)$ as described in this table in the two cases $3.$ and $4$. 

\noindent 
\begin{tabular}{|l|l|l|l|l|l|}
\hline
Gadget &\!\!I &\!\!II &\!\!III &\!\!IV &\!\!mismatches\\ 
\hline
$A(i)=$\!\! &\!\!$\varepsilon$ & \!\!$\mathtt{bba}A'(i_1)\triphash \! \ldots\! \triphash A'(i_d)$ &\!\!$\mathtt{bbb}X'\ \ \ \ \ \triphash\! \ldots\! \triphash X'$ &\!\!$\varepsilon$ & \\ 
$3.~B(j)=$\!\! &\!\!$x_j$ &\!\!$\mathtt{bba}B'(j_1)\triphash \! \ldots\! \triphash B'(j_d)$ &\!\!$y_j$ &\!\!$\varepsilon$ &\!\!$d$ (in II)\\  
$4.~B(j)=$\!\! &\!\!$\varepsilon$ &\!\!$x_j$ &\!\!$\mathtt{bba}B'(j_1)\triphash \! \ldots\! \triphash B'(j_d)$ &\!\!$y_j$ &\!\!$d+1$ (in IV)\\
\hline
\end{tabular}

Wrapping up, there are no other ways than those described in cases 1-6 above in which $B'(j)$ can be aligned to a factor of $w$. In particular, in order to reach an alignment with at most $n(d+1)-1$ mismatches, at least one $B'(j)$ should be aligned to a factor of $w$ such that it only causes $d$ mismatches (as in case 3). Thus, in that case we would have a pair of orthogonal vectors. Conversely, if there exist $u_i$ and $v_j$ which are orthogonal and $i\geq j$, then we can align $B'(j)$ to the occurrence of $\mathtt{bba}A'(i_1)\triphash  \ldots A'(i_d)$ from the first $A(i)$ and all the other gadgets $B'(\ell)$ to factors $\mathtt{bbb}X$, and obtain a number of $n(d+1)-1$ mismatches. Note that such an alignment is possible as there exist at least $j-1$ factors $\mathtt{bbb}X$ before the first $A(i)$ and at least $n$ more occurrences of $\mathtt{bbb}X$ after it; moreover the variables $x_\ell$ and $y_\ell$ can be used to align as desired the strings $B'(v_\ell)$ to the respective $\mathtt{bbb}X$ factors of $w$. If there exist $u_i$ and $v_j$ which are orthogonal and $i< j$, then we can align $B'(j)$ to the occurrence of $\mathtt{bba}A'(i_1)\triphash A'(i_2)\triphash \ldots A'(i_d)$ from the second $A(i)$ and all the other gadgets $B'(\ell)$ to factors $\mathtt{bbb}X$, and obtain again a number of $n(d+1)-1$ mismatches. This is possible for similar reasons to the ones described above. 


This shows that our reduction is correct. The instance of $\OV$ defined by $U$ and $V$ contains two orthogonal vectors if and only the instance of $\misMatch_{\regPat}$ defined by $w, \alpha, $ and $\Delta=n(d+1)-1$ can be answered positively. 
Moreover, the instance of $\misMatch_{\regPat}$ can be constructed in $O(nd)$ time and we have that $|w|,|\alpha|,\Delta \in \Theta(nd)$.

Assume now that there exists a solution of $\misMatch_{\regPat}$ running in $O(|w|^g|\alpha|^h)$ with $g+h=2-\epsilon$ for some $\epsilon<0$. This would lead to a solution for $\OV$ running in $O(nd + (nd)^{2-\epsilon})$, a contradiction to the $\OV$-conjecture. Similarlty, if there exists a solution of $\misMatch_{\regPat}$ running in $O(|w|^g\Delta ^h)$ with $g+h=2-\epsilon$ for some $\epsilon<0$, then there exists a solution for $\OV$ running in $O(nd + (nd)^{2-\epsilon})$, a contradiction to the $\OV$-conjecture. This proves our statement.
\end{proof}

\begin{remark}
An immediate consequence of the previous theorem is that $\minMisMatch_{\regPat}$ can not be solved in $\mathcal{O}(n^h \hdist{\alpha}{w}^g)$ time (or in $\mathcal{O}(|w|^h |\alpha|^g)$ time) with $h+g= 2-\epsilon$ for some $\epsilon>0$, unless the $\OV$-Conjecture fails. Thus, as $\hdist{\alpha}{w}\leq |\alpha|$, $\minMisMatch_{\regPat}$ and $\misMatch_{\regPat}$ cannot be solved polynomially faster than our algorithms, unless the $\OV$-Conjecture fails. 
\end{remark}

\section{Patterns with Repeated Variables}\label{sec:repeated}

In Section \ref{sec:Reg} we have shown that if no variable occurs more than once in the input pattern $\alpha$, then the problems $\misMatch$ and $\minMisMatch$ can be solved in polynomial time. Let us now consider patterns where variables are allowed to occur more than once, i.e., patterns with repeated variables. Firstly, we recall two measures of the structural complexity of patterns. 

For every variable $x\in \var(\alpha)$, the scope of $x$ in $\alpha$ is defined by $\mathtt{sc}_\alpha(x)=[i:j]$, where $i$ is the leftmost and $j$ the rightmost occurrence of $x$ in $\alpha$. The scopes of the variables $x_1,\ldots ,x_k\in \var(\alpha)$ coincide in $\alpha$ if $\cap_{i=1}^k \mathtt{sc}(x_i)\neq \emptyset$. By $\scd(\alpha)$ we denote the scope coincidence degree of $\alpha$: the maximum number of variables in $\alpha$ whose scopes coincide. By $\kScdPat$ we denote the class of patterns whose scope coincidence degree is at most $k$. 

Given a pattern $\alpha$, with $p$ variables, a marking sequence of $\alpha$ is an ordering $x_1<x_2<\ldots<x_p$ of $\var(\alpha)$. The skeleton $\alpha_{var}$ of $\alpha$ is obtained from $\alpha$ by removing all the terminals. A marking of $\alpha_{var}$ w.r.t. a marking sequence $x_1<x_2<\ldots<x_p$ of $\alpha$ is a $p$-steps procedure: in step $i$ we mark all occurrences of variable $x_i$. The pattern $\alpha $ is called $k$-local if and only if there exists a marking sequence of $x_1<x_2<\ldots<x_p$ of $\alpha$ such that, for $i$ from $1$ to $p$, the variables marked in the first $i$ steps of the marking of $\alpha_{var}$ w.r.t. this marking sequence form at most $k$ non-overlapping length-maximal factors in $\alpha_{var}$; the respective marking sequence is called witness for the $k$-locality of $\alpha$. By $\kLocPat$ we denote the class of $k$-local patterns. See \cite{DayFMN17,CaselDFKMS19} for an extended discussion and examples regarding $k$-locality.

Several more particular classes which we consider in this context are the following:
 \begin{itemize}
\item 
The class of unary patterns $\oneVarPat$: $\alpha \in \oneVarPat$ if there exists $x\in X$ such that $\var(\alpha)= \{x\}$;  example: $\alpha_1=\mathtt{ab} x \mathtt{ab} xx \mathtt{baab}\in \oneVarPat$.
\item 
The class of one-repeated-variable patterns $\oneRepPat$: $\alpha \in \oneRepPat$ if there exists at most one variable $x\in X$ such that $|\alpha|_x>1$; example: $\alpha_2=\mathtt{ab} x y \mathtt{ab}z xx \mathtt{baab} v\in \oneRepPat$.
\item The class $\nonCrossPat=\oneScdPat$, called the class of non-cross patterns; as examples, consider  $\alpha_3=\mathtt{ab} xx y \mathtt{ab}zzz \mathtt{bb} vvv \mathtt{ab} v u\in \nonCrossPat\setminus \oneRepPat$ and $\alpha_4=\mathtt{ab} x y \mathtt{ab}z xx \mathtt{bb} v \mathtt{ab} x\in \oneRepPat\setminus \nonCrossPat$. Note that $\alpha \in \nonCrossPat$ if and only if $\alpha$ can be written as the concatenation of several $\oneVarPat$-patterns, whose variables are pairwise distinct. Thus, $\nonCrossPat$-patterns are $1$-local. 
\end{itemize}

Note that in a $\nonCrossPat$-pattern $\alpha$, for any two variables $x,y\in \var(\alpha)$, where the last occurrence of $y$ is to the right of the first occurrence of $x$ in $\alpha$, we can actually write $\alpha = \beta x \gamma y \delta $ such that $x,y\notin \var(\gamma)$, $x\notin \var(\delta)$, and $y\notin \var(\beta)$. In other words, there are no interleaved occurrences of two variables. Moreover, if $\alpha\in \nonCrossPat$, then $\alpha$ is $1$-local: the marking sequence is obtained by ordering the variables according to the position of their first occurrence. 

Clearly, $\oneVarPat \subset \oneRepPat$ and $\oneVarPat \subset \nonCrossPat$, but $\oneRepPat$ and $\nonCrossPat$ are incomparable. Indeed, if $\alpha\in \nonCrossPat$ then $\alpha$ is $1$-local and $\oneRepPat$ contains patterns $\alpha$ with $\scd(\alpha)=2$. 

Now we briefly discuss the examples mentioned above.

Then, $\alpha_1=\mathtt{ab} x \mathtt{ab} xx \mathtt{baab}\in \oneVarPat$ ($x$ is the single variable). 

Secondly, $\alpha_2=\mathtt{ab} x y \mathtt{ab}z xx \mathtt{baab} v$, with $\var(\alpha_2)=\{x,y,z,v\}$, is in $\oneRepPat$ ($x$ is the repeated variable) but not in $\oneVarPat$ nor in $\nonCrossPat$, as $\scd(\alpha_2)=2$ and, more intuitively, the occurrences of $x$ are interleaved with those of the other variables. 

Then, $\alpha_3=\mathtt{ab} xx y \mathtt{ab}zzz \mathtt{bb} vvv \mathtt{ab} v u$, with $\var(\alpha_3)=\{x,y,z,v,u\}$, is in $\nonCrossPat$, but not in $\oneRepPat$ as each of $x,z,$ and $v$ occurs at least twice. 

Finally, $\alpha_4=\mathtt{ab} x y \mathtt{ab}z xx \mathtt{bb} v \mathtt{ab} x$ is in $\oneRepPat$ but it is not a non-cross pattern as $\scd(\alpha_4)=2$ and, for instance, we cannot write it as $\alpha_4 = \beta x \gamma v \delta $ such that $x,v\notin \var(\gamma)$, $x\notin \var(\delta)$, and $v\notin \var(\beta)$, i.e., we cannot separate the occurrences of the variables $x$ and $v$ -- they are interleaved. The pattern $\alpha_4$ is $2$-local, as witnessed, for instance, by the marking sequence $v<x<y<z$. 

Further, if $\alpha$ is a pattern and $x\in \var(\alpha)$, then an $x$-block is a factor $\alpha[i:j]$ such that 
$\alpha[i:j]\in \oneVarPat$ with $\var(\alpha[i:j])=x$ and it is length-maximal with this property: it cannot be extended to the right or to the left without introducing a variable different from $x$. 

The next lemma is fundamental for the results of this section.

\begin{restatable}{lemma}{lemmedian}\label{lem:median}
Given a set of words $w_1,\ldots,w_p\in \Sigma^m$, we can find in $O(|\Sigma|+mp)$ a {\em median string} for $\{w_1,\ldots,w_p\}$, i.e. a string $w$ such that $\sum_{j=1}^{p}\hdist{w_i}{w}$ is minimal.
\end{restatable}

\begin{proof}
We will use an array $C$ with $\Sigma$ elements, called counters, indexed by the letters of $\Sigma$, and all initially set to $0$. 
For each $i$ between $1$ and $m$, we count how many times each letter of $\Sigma$ occurs in the multi-set $\{w_1[i],w_2[i],\ldots, w_p[i]\}$ using $C$. Let $w[i]$ be the most frequent letter of this multi-set. After computing $w[i]$, we reset the counters which were changed in this iteration, and repeat the algorithm for $i+1$. After going through all values of $i$, we return the word $w=w[1]w[2]\ldots w[m]$ as the answer to the problem. 
The correctness of the algorithm is immediate, while its complexity is clearly $O(|\Sigma|+mp)$. 
\end{proof}

The typical use of this lemma is the following: we identify the factors of $w$ to which a repeated variable is aligned, and then compute the optimal assignment of this variable. Based on this, the following theorem can now be shown. 
\begin{restatable}{theorem}{thmOneVarPat}\label{thm:1VarPat}
$\minMisMatch_{\oneVarPat}$ and $\misMatch_{\oneVarPat}$ can be solved in $O(n)$ time.
\end{restatable}

\begin{proof}
It is enough to show how to solve $\minMisMatch_{\oneVarPat}$. 

Recall that we were given a word $w$, of length $n$, and a pattern $\alpha$, of length $m$. Let $x$ be the single variable that occurs in $\alpha$ and, for simplicity, we denote by $m_x$ the number of occurrences of $x$ in $\alpha$, i.e., $m_x=|\alpha|_x$. Thus, $\alpha=\prod_{i=1}^{m_x}(v_{i-1}x)v_{m_x}$, where $v_i\in \Sigma^*$ for all $i\in \{1,\ldots,m_x\}$.

Let $m'=m-m_x$ be the number of terminal symbols of $\alpha$. It is clear that $x$ should be mapped to a string of length $\ell=\frac{n-m'}{m_x}$. If $\ell$ is not an integer, there exists no string $u$ which can be obtained from $\alpha$ by substituting $x$ with a terminal-word such that $|u|=|w|$ and $\hdist{u}{w}$ is finite. So, let us assume $\ell$ is an integer. 

Now we know that we want to compute a string $u$ which can be obtained from $\alpha$ by substituting $x$ with a terminal-word $u_x$ of length exactly $\ell$. Moreover, $u=\prod_{i=1}^{m_x}(v_{i-1}u_x)v_{m_x}$. We define the factors $w_1,\ldots,w_{m_x}$ of $w$ such that $w_i=w[a_i+1:a_i+\ell_x]$ and $a_i=|\prod_{j=1}^{i-1}(v_{i-1}u_x)v_{i}|$. These are the factors that would align to the occurrences of $u_x$ when aligning $u$ with $w$. As the factors $v_i$ always create the same number of mismatches to the corresponding factors of $w$, irrespective on the choice of $u_x$, we need to choose $u_x$ such that $\sum_{j=1}^{m_x}\hdist{w_i}{u_x}$ is minimal. For this, we can use Lemma \ref{lem:median}, and compute $u_x$ in $O(|\Sigma|+m_x\ell_x)$ time. As it is our assumption that $|\Sigma|\leq n$, we immediately get that $u_x$ can be computed in $O(n)$ time. So $u$ can be computed in $O(n)$ time. To solve $\minMisMatch_{\oneVarPat}$, we simply return $\hdist{u}{w}$, and this can be again computed in linear time.
\end{proof}

By a standard dynamic programming approach, we use the previous result to obtain a polynomial-time solution for $\minMisMatch_{\nonCrossPat}$ based on the solution for $\minMisMatch_{\oneVarPat}$ (in the statement, $p=|\var(\alpha)|$). 
\begin{restatable}{theorem}{thmnonCross}\label{thm:nonCross}
$\minMisMatch_{\nonCrossPat}$ and $\misMatch_{\nonCrossPat}$ can be solved in $O(n^3p)$~time.
\end{restatable}

\begin{proof}
It is enough to show how to solve $\minMisMatch_{\nonCrossPat}$. Once more, we were given a word $w$, of length $n$, and a pattern $\alpha$, of length $m$. Assume $\var(\alpha)=\{x_1,\ldots,x_p\}$, and we have $\alpha=\beta_1\beta_2\cdots\beta_p$, where $\beta_{2i+1}$ is an $x_{2i+1}$-block, for all $i$ such that $1\leq 2i+1\leq m$, and $\var(\beta_{2i})=\{x_{2i}\}$, for all $i$ such that $1< 2i\leq m$. Let $\alpha_\ell = \beta_1\cdots \beta_\ell$, for $\ell\geq 1$. 

The idea of our algorithm is the following.

For $\ell$ from $1$ to $p$, we define $Dist[j][\ell]=\hdist{\alpha_\ell}{w[1:j]}$ for all prefixes $w[1:j]$ of $w$. This matrix can be computed by dynamic programming.

For $\ell=1$, we can use Theorem \ref{thm:1VarPat} to compute each element $Dist[j][1]$ in linear time. So, $Dist[\cdot][1]$ is computed in $O(n^2)$ time. 

Consider now the case when $\ell>1$ and assume we have computed the array $Dist[\cdot][\ell-1]$. For a position $j$ of the word $w$, we compute $Dist[j][\ell]=\min\{Dist[j'][\ell-1]+\hdist{\beta_\ell}{w[j'+1:j]}\mid j'\leq j\}$, where $\hdist{\beta_\ell}{w[j'+1:j]}$ is computed, once more, by Theorem \ref{thm:1VarPat}. It is clear that computing each element $Dist[j][\ell]$ as described above is correct, and that this computation takes $O(n^2)$ time. 

Therefore, we can compute all elements of the matrix $Dist[\cdot][\cdot]$ in $O(n^3 p)$ time. We return $Dist[n][p]$ as the answer to $\minMisMatch_{\nonCrossPat}$. 
\end{proof}

The results presented so far show that $\minMisMatch_P$ and $\misMatch_P$ can be solved in polynomial time, as long as we do not allow interleaved occurrences of variables in the patterns of the class $P$. We now consider the case of $\oneRepPat$-patterns, the simplest class of patterns which permits interleaved occurrences of variables. 

For simplicity, in the results regarding $\oneRepPat$ we assume that the variable which occurs more than once in the input pattern is denoted by $x$.
\begin{restatable}{theorem}{thmoneRep}\label{thm:oneRep}
$\minMisMatch_{\oneRepPat}$ and $\misMatch_{\oneRepPat}$ can be solved in $O(n^{k+2}m)$ time, where $k$ is the number of $x$-blocks in the input pattern $\alpha$.
\end{restatable}

\begin{proof}
Once more, we only show how $\minMisMatch_{\oneRepPat}$ can be solved. The result for $\misMatch_{\oneRepPat}$ follows then immediately. 

In $\minMisMatch_{\oneRepPat}$, we are given a word $w$, of length $n$, and a pattern $\alpha$, of length $m$, which, as stated above, has exactly $k$ $x$-blocks. Thus $\alpha=\prod_{i=1}^k(\gamma_{i-1}\beta_i)\gamma_k$, where the factors $\beta_i$, for $i\in\{1,\ldots,k\}$, are the $x$-blocks of $\alpha$. It is easy to observe that $\var(\gamma_i)\cap \var(\gamma_j)=\emptyset$, for all $i$ and $j$, and  $\gamma=\gamma_0\gamma_1\cdots \gamma_k$ is a regular pattern. 

When aligning $\alpha$ to $w$ we actually align each of the patterns $\gamma_j$ and $\beta_i$, for $0\leq j\leq k$ and $1 \leq i\leq k$, to respective factors of the word $w$. Moreover, the factors to which these patterns are respectively aligned are completely determined by the length $\ell$ of the image of $x$, and the starting positions $h_i$ of the factors aligned to the patterns $\beta_i$, for $1\leq i\leq k$. Knowing the length $\ell$ of the image of $x$, we can also compute, for $1\leq i\leq k$, the length $\ell_i$ of $\beta_i$, when $x$ is replaced by a string of length $\ell$. In this case, $\gamma_0$ is aligned $u_0=w[1..h_1-1]$ and, for $1\leq i\leq k$, $\beta_i$ is aligned to $w_i=w[h_i:h_i+\ell_i-1]$ and $\gamma_i$ is aligned $u_i=w[h_{i-1}+\ell_{i-1}:h_i-1]$. Thus, $\beta_1\cdots\beta_k$ matches $w_1\cdots w_k$ and we can use Theorem \ref{thm:1VarPat} to determine $\hdist{\beta_1\cdots\beta_k}{w_1\cdots w_k}$ (or, in other words, determine the string $u_x$ that should replace $x$ in order to realize this Hamming distance). Further, we can use Theorem \ref{thm:minMisMatch} to compute $\hdist{\gamma_i}{u_i}$, for all $i\in \{0,\ldots,k\}$. Adding all these distances up, we obtain a total distance $D_{\ell,h_1,\ldots,h_k}$; this value depends on $\ell,h_1,\ldots,h_k$. 

So, we can simply iterate over all possible choices for $\ell,h_1,\ldots,h_k$ and find $\hdist{\alpha}{w}$ as the minimum of the numbers $D_{\ell,h_1,\ldots,h_k}$. 

By the explanations above, it is straightforward that the approach is correct: we simply try all possibilities of aligning $\alpha$ with $w$. The time complexity is, for each choice of $\ell,h_1,\ldots,h_k$, $O(\sum_{i=1}^k |w_i|)\subseteq O(n)$ for the part corresponding to the computation of the optimal alignment between the factors $\beta_i$ and the words $w_i$, and $O(\sum_{i=0}^k |u_i|\hdist{\gamma_i}{u_i})\subseteq O(nm)$ for the part corresponding to the computation of the optimal alignment between the factors $\gamma_i$ and the words $u_i$. So, the overall complexity of this algorithm is $O(n^{k+2}m)$. 
\end{proof}

We can also show the following more general result.
\begin{restatable}{theorem}{thmkLoc}\label{thm:kLoc}
$\minMisMatch_{\kLocPat}$ and $\misMatch_{\kLocPat}$ can be solved in $O(n^{2k+2}m)$ time.
\end{restatable}

\begin{proof}
We only present the solution for $\minMisMatch_{\kLocPat}$ (as it trivially works in the case of $\misMatch_{\kLocPat}$ too).

Let us note that, by the results in \cite{DayFMN17}, we can compute a marking sequence of $\alpha$ in $O(m^{2k}k)$ time. So, after such a preprocessing phase, we can assume that we have a word $w$, a $k$-local pattern $\alpha$ (with $p$ variables) with a witness marking sequence $x_1\leq \ldots \leq x_p$ for the $k$-locality of $\alpha$, and we want to compute $\hdist{\alpha}{w}$.

Generally, the main idea behind matching $\kLocPat$-patterns is that when looking for possible ways to align such a pattern $\alpha$ to a word $w$ we can consider the variables in the order given by the marking sequence, and, when reaching variable $x_i$, we try all possible assignments for $x_i$. The critical observation here is that after each such assignment of a new variable, we only need to keep track of the way the $t\leq k$ length-maximal factors of $\alpha$, which contain only marked variables and terminals, match (at most) $t\leq k$ factors of $w$. 

We will use this approach in our algorithm for $\minMisMatch_{\kLocPat}$. 

The first step of this algorithm is the following. We go through $\alpha$ and identify all $x_1$-blocks: $\beta_{1,1},\ldots,\beta_{1,j_1}$. Because $\alpha$ is $k$-local, we have  that $j_1\leq k$. For each $2j_1$-tuple $(i_1,\ldots,i_{2j_1})$ of positions of $w$, we compute the minimum number of mismatches if we align (simultaneously) the patterns $\beta_g$ to the factors $w[i_{2g-1}:i_{2g}]$, for $g$ from $1$ to $j_1$, respectively. This reduces to finding an assignment for $x_1$ which aligns optimally the patterns $\beta_{1,g}$ to the respective factors, and can be done in $O(n)$ time using Theorem \ref{thm:1VarPat}. For each $2j_1$-tuple $(i_1,\ldots,i_{2j_1})$ of positions of $w$, we denote by $M_1(i_1,\ldots,i_{2j_1})$ the minimum number of mismatches resulting from the (simultaneous) alignment of the patterns $\beta_{1,g}$ to the factors $w[i_{2g-1}:i_{2g}]$, for $g$ from $1$ to $j_1$, respectively. Clearly, $M_1$ can be seen as a $j_1$-dimensional array. 

Assume that after $h\geq 1$ steps of our algorithm we have computed the factors $\beta_{h,1},\ldots,\beta_{h,j_h}$ of $\alpha$, which are length-maximal factors of $\alpha$ which only contain the variables $x_1,\ldots,x_h$ and terminals (i.e., extending them to the left or right would introduce a new variable $x_\ell$ with $\ell>h$); as $\alpha$ is $k$-local, we have $j_h\leq k$. Moreover, for each $2j_h$-tuple $(i_1,\ldots,i_{2j_h})$ of positions of $w$, we have computed $M_h(i_1,\ldots,i_{2j_h})$, the minimum number of mismatches if we align (simultaneously) the patterns $\beta_{h,g}$ to the factors $w[i_{2g-1}:i_{2g}]$, for $g$ from $1$ to $j_h$, respectively. $M_h$ is implemented as a $j_h$ dimensional array, and this assumption clearly holds after the first step. 

We now explain how step $h+1$ is performed. 
\begin{enumerate}

    \item We compute the factors $\beta_{h+1,1},\ldots,\beta_{h+1,j_{h+1}}$ of $\alpha$, which are length-maximal factors of $\alpha$ which only contain the variables $x_1,\ldots,x_{h+1}$ and terminals (i.e., extending them to the left or right would introduce a new variable $x_\ell$ with $\ell>h+1$). Clearly, $\beta_{h+1,r}$ is either an $x_{h+1}$-block or it has the form $\beta_{h+1,r}=\gamma_{r,0}\beta_{h,a_r}\gamma_{r,1}\cdots \beta_{r,a_r+b_r}\gamma_{r,b_r+1}$ where the patterns $\gamma_{r,t}$ contain only the variable $x_{h+1}$ and terminals and extending $\beta_{h+1,r}$ to the left or right would introduce a new variable $x_\ell$ with $\ell>h+1$.
    \item We initialize the values $M_{h+1}(i_1,\ldots,i_{2j_{h+1}})\gets \infty$, for each $2j_{h+1}$-tuple $(i_1,\ldots,i_{2j_{h+1}})$ of positions of $w$.  
    \item For each $\ell\leq n$ (where $\ell$ corresponds to the length of the image of $x_{h+1}$) and each $2j_{h}$-tuple $(i_1,\ldots,i_{2j_{h}})$ of positions of $w$ such that $M_h(i_1,\ldots,i_{2j_{h}})$ is finite do the following:
    \begin{enumerate}
        \item We compute the tuple $(i'_1,\ldots,i'_{2j_{h+1}})$ such that $\beta_{h+1,g}$ is aligned to the factor $w[i'_{2g-1}:i'_{2g}]$, for $g$ from $1$ to $j_{h+1}$, respectively. This can be computed based on the fact that the factors $\beta_{h,g}$ are aligned to the factors $w[i_{2g-1}:i_{2g}]$, for $g$ from $1$ to $j_h$, respectively, and the image of $x_{h+1}$ has length $\ell$. 
        \item We compute the factors of $w$ aligned to $x_{h+1}$ in the alignment computed in the previous line. Then, we can use the algorithm from Theorem \ref{thm:1VarPat} and the value of $M_h(i_1,\ldots,i_{2j_{h}})$ to compute an assignment for $x_{h+1}$ which aligns optimally the patterns $\beta_{h+1,g}$ to the corresponding factors of $w$.
        \item If the number of the mismatches in this alignment is smaller than the current value of $M_{h+1}(i'_1,\ldots,i'_{2j_{h+1}})$, we update $M_{h+1}(i'_1,\ldots,i'_{2j_{h+1}})$. 
    \end{enumerate}
\end{enumerate}

This dynamic programming approach is clearly correct. In $M_{h+1}(i_1,\ldots,i_{2j_{h+1}})$ we have the optimal alignment of the patterns $\beta_{h+1,1}, \ldots, \beta_{h+1,j_{h+1}}$ to $w[i_1:i_2], \ldots, w[i_{2j_{h+1}-1}:i_{2j_{h+1}}]$. 
As far as the complexity is concerned, the lines $1$, $3.a,$ $3.b$, $3.c$ can be implemented in linear time, while the for-loop is iterated $O(n^{2k+1})$ times. Line $2$ takes $O(n^{2k})$ times. The whole computation in step $h+1$ of the algorithm takes, thus, $O(n^{2k+1})$ time.

Now, we execute the procedure described above for $h$ from $2$ to $m$, and, in the end, we compute the array $M_{m}$. The answer to our instance of the problem $\minMisMatch_{\kLocPat}$ is $M_m(1,n)$. The overall time complexity needed to perform this computation is $O(mn^{2k+1})$ time.
\end{proof}

Note that $\nonCrossPat$-patterns are $1$-local, while the locality of an $\oneRepPat$-pattern is upper bounded by the number of $x$-blocks. However, the algorithms we obtained in those particular cases are more efficient than the ones which follow from Theorem \ref{thm:kLoc}. 

The fact that Lemma \ref{lem:median} is used as the main building block for our results regarding $\misMatch_{P}$ and $\minMisMatch_{P}$ for $P\in\{\oneRepPat,\kLocPat\}$, suggests that these problems could be closely related to the following well-studied problem \cite{similarRegions,FellowsGN06,BoucherEPTAS,BulteauS20}.
\begin{problem}
  \problemtitle{Consensus Patterns: $\CP$}
  \probleminput{$k$ strings $w_1, \ldots, w_k \in \sig{\ell}$, integer $m \in \mathbb{N}$ with $m\leq\ell$, an integer $\Delta\leq mk$.}
  \problemquestion{Do the strings $s$, of length $m$, and $s_1, \ldots, s_k$, factors of length $m$ of each $w_1, \ldots, w_k$, respectively, exist, such that $\sum_{i=1}^k \hdist{s_i}{s}\leq \Delta$?}
\end{problem}

Exploiting this connection, and following the ideas of \cite{similarRegions}, we can show the following theorem. In this theorem we restrict to the case when the input word $w$ of $\minMisMatch_{\oneRepPat}$ is over  $\Sigma=\{1,\ldots,\sigma\}$ of constant size $\sigma $. 
\begin{restatable}{theorem}{thmOnerepPTAS}\label{thm:1repPTAS}
For each constant $r\geq 3$, there exists an algorithm with run-time $O(n^{r+3})$ for $\minMisMatch_{\oneRepPat}$ whose output distance is at most $\min\left\{2,\left(1+\frac{4\sigma-4}{\sqrt{e}(\sqrt{4r+1}-3)}\right)\right\}\hdist{\alpha}{w}$. 
\end{restatable}

\begin{proof}
We first note that there exists a relatively simple algorithm solving $\minMisMatch_{\oneRepPat}$ such that the output distance is no more than $2\hdist{\alpha}{w}$ (which also works for integer alphabets). 

Indeed, assume that we have a substitution $h$ for which $\hdist{h(\alpha)}{w}=\hdist{\alpha}{w}$. Assume that the repeated variable $x$ is mapped by $h$ to a string $u$ and the $t$ occurrences of $x$ are aligned, under $h$, to the factors $w_1,w_2,\ldots, w_t$ of $w$. Now, let $w_i$ be such $\hdist{u}{w_i}\leq \hdist{u}{w_j}$ for all $j\neq i$. Let us consider now the substitution $h'$ which substitutes $x$ by $w_i$ and all the other variables exactly as $h$ did. We claim that $\hdist{h'(\alpha)}{u}\leq 2 \hdist{h(\alpha)}{u}$. It is easy to see that $\hdist{h'(\alpha)}{w}-\hdist{h(\alpha)}{w}= \sum_{j=i}^{t}(\hdist{w_i}{w_j} - \hdist{u}{w_j})\leq \sum_{j=i}^{t}(\hdist{w_i}{u} + \hdist{u}{w_j} - \hdist{u}{w_i})$ (where the last inequality follows from the triangle inequality for the Hamming Distance). 
Thus, $\hdist{h'(\alpha)}{w}-\hdist{h(\alpha)}{w}\leq \sum_{j=i}^{t}\hdist{w_i}{u}\leq \sum_{j=i}^{t}\hdist{w_j}{u}\leq \hdist{h(\alpha)}{u}.$ So our claim holds.

A consequence of the previous observation is that there exists a substitution $h'$ that maps $x$ to a factor of $w$ and produces a string $h'(\alpha)$ such that $\hdist{h'(\alpha)}{u}\leq 2 \hdist{\alpha}{u}$. So, for each factor $u$ of $w$, we $x$ by $u$ in $\alpha$ to obtain a regular pattern $\alpha'$, then use Theorem \ref{thm:minMisMatch} to compute $\hdist{\alpha'}{w}$. We return the smallest value $\hdist{\alpha'}{w}$ achieved in this way. Clearly, this is at most $2 \hdist{\alpha}{u}$. The complexity of this algorithm is $O(n^4)$, as it simply uses the quadratic algorithm of Theorem \ref{thm:minMisMatch} for each factor of $w$. 

We will now show how this algorithm can be modified to produce a value closer to $\hdist{\alpha}{w}$, while being less efficient. 

The algorithm consists of the following main steps:
\begin{enumerate}
    \item For $\ell\leq n/r$ and $r$ factors $u_1,\ldots,u_r$ of length $\ell$ of $w$ do the following:
    \begin{enumerate}
        \item Compute $u_{u_1,\ldots,u_r}$ the median string of $u_1,\ldots,u_r$ using Lemma \ref{lem:median}.
        \item Let $\alpha'$ be the regular pattern obtained by replacing $x$ by $u_{u_1,\ldots,u_r}$ in $\alpha$.
        \item Compute the distance $d_{u_1,\ldots,u_r}=\hdist{\alpha'}{w}$ using Theorem \ref{thm:minMisMatch}.
    \end{enumerate}
    \item Return the smallest distance $d_{u_1,\ldots,u_r}$ computed in the loop above.
\end{enumerate}

Clearly, for $r=1$ the above algorithm corresponds to the simple algorithm presented in the beginning of this proof. Let us analyse its performance for an arbitrary choice of $r$. 

The complexity is easy to compute: we need to consider all possible choices for $\ell$ and the starting positions of $u_1,\ldots,u_r$. So, we have $O(n^{r+1})$ possibilities to select the non-overlapping factors $u_1,\ldots,u_r$ of length $\ell $ of $w$. The computation done inside the loop can be performed in $O(n^2)$ time. So, overall, our algorithm runs in $O(n^{r+3})$ time. 

Now, we want to estimate how far away from $\hdist{\alpha}{w}$ is the value this algorithm returns. In this case, we will make use of the fact that the input terminal-alphabet is constant. We follow closely (and adapt to our setting) the approach from \cite{similarRegions}.

Firstly, a notation. In step 1.b of the algorithm above, we align $\alpha'$ to $w$ with a minimal number of mismatches. In this alignment, let $d'_{u_1,\ldots,u_r}$ be the total number of mismatches caused by the factors $u_{u_1,\ldots,u_r}$ which replaced the occurrences of the variable $x$ in $\alpha$.

Now, assume that we have a substitution $h$ for which $\hdist{h(\alpha)}{w}=\hdist{\alpha}{w}=d_{opt}$. Assume also that the repeated variable $x$ is mapped by $h$ to a string $u_{opt}$ of length $L$ and the $t$ occurrences of $x$ are aligned, under $h$, to the factors $w_1,w_2,\ldots, w_t$ of $w$. Let $d'_{opt}$ be the number of mismatches caused by the alignment of the images of the $t$ occurrences of $x$ under $h$ to the factors $w_1,w_2,\ldots, w_t$. Finally, let $\rho=1+\frac{4\sigma-4}{\sqrt{e}(\sqrt{4r+1}-3)}.$

Note that, for $\ell=L$, $u_1,\ldots, u_r$ correspond to a set of randomly chosen numbers $i_1,\ldots,i_r$ from $\{1,\ldots,n\}$: their starting positions. We will show in the following that $E\left[d'_{u_1,\ldots,u_r}\right]\leq \rho d'_{opt}$. If this inequality holds, then we can apply the probabilistic method: there exists at least a choice of $u_1,\ldots,u_r$ of length $L$ such that $d'_{u_1,\ldots,u_r} \leq \rho d'_{opt}$. As we try all possible lengths $\ell$ and all variants for choosing $u_1,\ldots,u_r$ of length $\ell$, we will also consider the choice of $u_1,\ldots,u_r$ of length $L$ such that $d'_{u_1,\ldots,u_r} \leq \rho d'_{opt}$, and it is immediate that, for that, for the respective $u_1,\ldots,u_r$  we also have that $d_{u_1,\ldots,u_r}\leq \rho d_{opt}$. Thus, the value returned by our algorithm is at most $\rho d_{opt}$. 

So, let us show the inequality $E\left[d'_{u_1,\ldots,u_r}\right]\leq \rho d_{opt}$. 

For $\mathtt{a}\in \Sigma$, let $f_j(\mathtt{a})=|\{i\mid 1\leq i\leq t, w_i[j]=\mathtt{a}\}|$. Now, for an arbitrary string $s$ of length $L$, we have that $\sum_{i=1}^t\hdist{w_i}{s} = \sum_{j=1}^L(t-f_j(s[j]))$. So, for $s=u_{opt}$ we get $\sum_{i=1}^t\hdist{w_i}{u_{opt}} = \sum_{j=1}^L(t-f_j(u_{opt}[j]))$, and for $s=u_{u_1,\ldots,u_r}$ we have that $d'_{opt}=\sum_{i=1}^t\hdist{w_i}{u_{u_1,\ldots,u_r}} = \sum_{i=j}^L(t-f_j(u_{u_1,\ldots,u_r}[j]))$. 

Therefore, $E\left[d'_{u_1,\ldots,u_r}\right]=E\left[\sum_{j=1}^L(t-f_j(u_{u_1,\ldots,u_r}[j]))\right]=\sum_{j=1}^LE\left[t-f_j(u_{u_1,\ldots,u_r}[j])\right]$. 

Consequently, $E\left[d'_{u_1,\ldots,u_r}-d'_{opt}\right]=\sum_{j=1}^L(E\left[t-f_j(u_{u_1,\ldots,u_r}[j])\right]-t+f_j(u_{opt}[j])).$

That is, $E\left[d'_{u_1,\ldots,u_r}-d'_{opt}\right]=\sum_{j=1}^LE\left[f_j(u_{opt}[j])-f_j(u_{u_1,\ldots,u_r}[j])\right].$

By Lemma 7 of \cite{similarRegions}, we have that $E\left[f_j(u_{opt}[j])-f_j(u_{u_1,\ldots,u_r}[j])\right]\leq (\rho-1)(t-f_j(u_{opt}[j])).$

Hence, $E\left[d'_{u_1,\ldots,u_r}-d'_{opt}\right]\leq (\rho-1)\sum_{j=1}^L(t-f_j(u_{opt}[j]))=(\rho-1)d'_{opt}.$

So, we indeed have that $E\left[d'_{u_1,\ldots,u_r}\right]\leq \rho d'_{opt}.$

In conclusion, the statement of the theorem holds.
\end{proof}

It remains open whether other algorithmic results related to $\CP$ (such as those from, e.g.,  \cite{BrejovaBHLV05,BrejovaBHV06,Marx08}) apply to our setting too. 

In the following we show two hardness results which explain why the algorithms in Theorems \ref{thm:oneRep} and \ref{thm:1repPTAS} are interesting.

\begin{restatable}{theorem}{thmoneRepWOne}\label{thm:oneRepW1}
$\misMatch_\oneRepPat$ is $W[1]$-hard w.r.t. the number of $x$-blocks. 
\end{restatable}

\begin{proof}
We reduce $\CP$ to $\misMatch_\oneRepPat$, such that an instance of $\CP$ with $k$ different input strings is mapped to an instance of $\misMatch_\oneRepPat$ with $k+1$ $x$-blocks (where $x$ is the repeated variable), each containing exactly one occurrence of $x$. 

Hence, we consider an instance of $\CP$ which consists of $k$  strings $w_1, \ldots w_k \in \sig{\ell}$ of length $\ell$ and two integer $m, \Delta$ defining the length of the target factors and the number of allowed mismatches, respectively. 

The instance of $\misMatch_\oneRepPat$ which we construct consists of a text $w$ and a pattern $\alpha$, such that $\alpha$ contains $k+1$ $x$-blocks, each with exactly one occurrence of $x$, and is of polynomial size w.r.t. the size of the $\CP$-instance. Moreover, the number of mismatches allowed in this instance of $\misMatch_\oneRepPat$ is $\Delta' = m+\Delta$. That is, if there exists a solution for the $\CP$-instance with $\Delta$ allowed mismatches, then, and only then, we should be able to find a solution of the $\misMatch_\oneRepPat$-instance with $\Delta+m$ mismatches. 

The construction of the $\minMisMatch_{\oneRepPat}$ is realized in such a way that the word $w$ encodes the input strings, while $\alpha$ creates the mechanism for selecting the string $s$ and corresponding factors $s_1,\ldots,s_k$. The general idea is that $x$ should be mapped to $s$, and the factors to which the occurrences of $x$ are aligned should correspond to the strings $s_1,\ldots,s_k$. 

The structure of the word $w$ and that of the pattern $\alpha$ ensure that, in an alignment of $\alpha$ with $w$ which cannot be traced back to a admissible solution for the $\CP$-instance (that is, the occurrences of $x$ are not aligned to factors of length $m$ of the words $w_1,\ldots,w_k$ or $x$ is not mapped to a string of length $m$) we have at least $M \gg \Delta'$ mismatches, hence it cannot lead to a positive answer for the constructed instance of $\misMatch_{\oneRepPat}$. 

The reduction consists of three main steps. Firstly, we present a pair of gadgets to encode the relation of the strings $w_i$ and their factors $s_i$, for $i$ from $1$ to $k$. Then, we present a second pair of gadgets, which ensures that, in a positive solution of $\misMatch_{\oneRepPat}$, the variable $x$ can only be mapped to a string of length $m$, corresponding to the string $s$.
Finally, we show how to assemble these gadgets into the input word $w$ and the input pattern $\alpha$ for $\misMatch_\oneRepPat$.

\textbf{First pair of gadgets.}
We introduce the new letters $\mathtt{\set{a,b}}$, not contained in the input alphabet of the $\CP$-instance, as well as the variable $x$ and two fresh variables $y_i, z_i$, for each $i$ form $1$ to $k$. We construct the following two gadgets for each input string $w_i$ with $1 \leq i \leq k$. 
\begin{itemize}
    \item A gadget to be included in $w$: $\mathtt{g_i = w_i \overbrace{a^Mb^M\ldots a^Mb^M}^M}$.
    \item A gadget to be included in $\alpha$: $\mathtt{f_i = y_ixz_i \overbrace{a^Mb^M\ldots a^Mb^M}^M}$.
\end{itemize}
These gadgets allows us to align the $i^{th}$ occurrence of $x$ to an arbitrary factor of the word $w_i$, for $i$ from $1$ to $k$. 

\textbf{Second pair of gadgets.}
In this case, we use three new letters $\mathtt{\set{c,d,\$}}$ which are not contained in the input alphabet of $\CP$. Also, let $M=(k\ell)^2$. We define two new gadgets.
\begin{itemize}
    \item A gadget to be included in $w$: $\mathtt{A_w=\overbrace{c^Md^M\ldots c^Md^M}^M}\$^m$.
    \item A gadget to be included in $\alpha$: $\mathtt{A_{\alpha}=\overbrace{c^Md^M\ldots c^Md^M}^M}x$.
\end{itemize}
These gadgets enforce that, in an alignment of $\alpha$ and $w$, the variable $x$ is mapped to a string of length $m$, at the cost of exactly $m$ extra mismatches. Note that, because $\Delta\leq km$, we have that $M\gg \Delta$. 

\textbf{Final assemblage.} The word $w$ and the pattern $\alpha$ are defined as follows.
\begin{itemize}
    \item $w= \mathtt{ g_1g_2\ldots g_k A_w}$ and $\alpha = \mathtt{ f_1f_2\ldots f_k A_{\alpha}}$.
\end{itemize}

To wrap up, the instance of $\minMisMatch_{\oneRepPat}$ is defined by $w,\alpha,\Delta+m$.

{\bf The correctness of the reduction.} We will show that our reduction is correct by a detailed case analysis. We consider an alignment of $\alpha$ and $w$ with minimal number of mismatches, and we make the following observations. 
\begin{itemize}
    \item[A.] Firstly, if every $g_i$ is aligned to $f_i$, for $i$ from $i$ to $k$, it is immediate that $x$ is mapped to a string of length $m$, as the last occurrence of $x$ will be aligned to the $\$^m $ suffix of $w$. Thus, the total number of mismatches between $\alpha$ and $w$ in an alignment with a minimum number of mismatches is upper-bounded by $(k+1)m$. 
    \item[B.] Secondly, we assume, for the sake of a contradiction, that the length of the image of $x$ is not $m$. If $|x|>m$ (respectively, $|x|<m$) then the prefix $\mathtt{{(c^Md^M)}^M}$ of $A_\alpha$ is aligned to a factor of $w$ which starts strictly to the left of (respectively, to the right of) the first position of the prefix $\mathtt{{(c^Md^M)}^M}$ of $A_w$. It is not hard to see that this causes at least $M$ mismatches. Indeed, in the case when $|x|>m$, if the factor $\mathtt{{(c^Md^M)}^M}$ of $\alpha$ is aligned to a factor that starts at least $M$ position to the left of the factor $\mathtt{{(c^Md^M)}^M}$ of $w$, the conclusion is immediate; if the factor $\mathtt{{(c^Md^M)}^M}$ starts less then $M$ positions to the left of the factor $\mathtt{{(c^Md^M)}^M}$ of $w$, then each group $c^M$ in $\alpha$ will be aligned to a factor of $w$ that includes at least a $\mathtt{d}$ letter, so we again reach the conclusion. In the case when $|x|<m$, then, again, each group $c^M$ in $\alpha$ will be aligned to a factor of $w$ that includes at least a $\mathtt{d}$ letter, so the alignment leads to at least $M$ mismatches. \\
    So, we can assume from now on that $x$ is mapped to a string of length $m$. This also implies that $A_\alpha$ and $A_w$ are aligned, so we will largely neglect them from now on.
    \item[C.] Thirdly, we assume that there exists $i$ such that $\len{h(y_i)} + \len{h(z_i)} \neq \len{w_i}-m$. Let $j=\min\{i\leq k \mid \len{h(y_i)} + \len{h(z_i)} \neq \len{w_i}-m\}$. Then the suffixes $\mathtt{{(a^Mb^M)}^M}$ of $g_j$ and $f_j$ do not align perfectly to each other. If  $\len{h(y_j)} + \len{h(z_j)} < \len{w_i}-m$, then the suffix $\mathtt{{(a^Mb^M)}^M}$ of $f_j$ is aligned to a factor of $w$ which starts inside $w_j$. This immediately causes at least $M$ mismatches, as each group $\mathtt{a^M}$ will overlap to a group of which contains at least one $\mathtt{b}$ letter. If  $\len{h(y_j)} + \len{h(z_j)} > \len{w_i}-m$, then the suffix $\mathtt{{(a^Mb^M)}^M}$ of $f_j$ is aligned to a factor of $w$ which starts strictly to the right of the factor $w_j$. However, because $M=(k\ell)^2\gg k\ell$, and $f_j$ and $g_j$ are followed by the same number of factors $\mathtt{{(a^Mb^M)}^M}$ (until the factors $A_\alpha$ and $A_w$ are reached), the factor corresponding to the suffix $\mathtt{{(a^Mb^M)}^M}$ of $f_j$ cannot start more than $k\ell$ positions to the right of $w_j$. It is then immediate that this factor $\mathtt{{(a^Mb^M)}^M}$ of $f_j$ will cause at least $M$ mismatches: each group $\mathtt{a^M}$ will overlap to a group of which contains at least one $\mathtt{b}$ letter. \\
    So, from now on we can assume that the factors $\mathtt{{(a^Mb^M)}^M}$ of $g_j$ and $f_j$ are aligned. 
    \item[D.] At this point, it is clear that in each alignment of $\alpha$ and $w$ which fulfils the conditions described in items B and C: the variable $x$ is mapped to a string of length $m$, and its first $k$ occurrences are aligned to factors of the words $w_1,\ldots,w_k$. We will now show that for each alignment of $\alpha$ and $w$ in which the image of $x$ contains a $\$ $ symbol and fulfills the  conditions above, there exists an alignment of $\alpha$ and $w$ with at most the same number of mismatches, in which the image of $x$ does not contain a $\$ $ symbol and, once more, fulfills the conditions B and C. Assume that in our original alignment $x$ is mapped to a string $u_x$ of length $m$ such that $u_x[i]=\$ $. Let $u_1,\ldots,u_k$ be the factors of $w_1,\ldots,w_k$, respectively, to which the first occurrences of the variable $x$ are aligned. Consider the string $u'_x$ which is obtained from $u_x$ by simply replacing the $\$ $ symbol on position $i$ by $u_1[i]$. And then consider the alignment of $\alpha$ and $w$ which is obtained from the original alignment by changing the image of $x$ to $u'_x$ instead of $u_x$. When compared to the original alignment, the new alignment has an additional mismatch caused by the occurrence of $x$ aligned to $\$^m$, but at least one less mismatch caused by the alignments of the first $k$ occurrences of $x$. Indeed, in the original alignment, the $i^{th}$ position of $u_x$ was a mismatch to the $i^{th}$ position of any string $u_1,\ldots, u_k$, but now at least the $i^{th}$ positions of $w_1$ and $u'_x$ coincide. This shows that our claim holds. A similar argument shows that for any alignment in which $x$ is mapped to a string containing other letters than the input letters from the $\CP$-instance there exits an alignment in which $x$ is mapped to a string containing only letters from the $\CP$-instance.\\
    Hence, from now on we can assume that the factors $\mathtt{{(a^Mb^M)}^M}$ of $g_j$ and $f_j$ are aligned and that the image of $x$ has length $m$ and is over the input alphabet of $\CP$-instance. 
\end{itemize}

Based on the observations A-D, we can show that the reduction has the desired~properties. If the $\CP$-instance admits a solution $s,s_1,\ldots,s_k$ which causes a number of mismatches less or equal to $\Delta$, then we can produce an alignment of $\alpha$ to $w$ as follows. We map $x$ to $s$ and, for $i$ from $1$ to $k$, we map $x_i$ and $y_i$ to the prefix of $w_i$ occurring before $s_i$ and, respectively, the suffix of $w_i$ occurring after $s_i$. This leads to $\Delta+m$ mismatches between $\alpha$ and $w$, so the input $(w,\alpha,\Delta+m)$ of $\misMatch_{\oneRepPat}$ is accepted. Conversely, if we have an alignment of $\alpha$ and $w$ with at most $\Delta+m$ mismatches, then we have an alignment with the same number of mismatches which fulfills the conditions summarized at the end of item D above. Hence, we can define $s$ as the image of $x$ in this alignment, and the strings $s_1,\ldots,s_k$ as the factors of $w$ aligned to the first $k$ occurrences of $x$ from $\alpha$. Clearly, for $i$ between $1$ and $k$, $s_i$ is a factor of $w_i$. As $m$ mismatches of the alignment were caused by the alignment of the last $x$ to $\$^m$, we get that $\sum_{i=1}^k\hdist{s}{s_i}\leq \Delta$. Thus, the instance of $\CP$ is accepted. 

This concludes the proof of the correctness of our reduction. As $M$ is clearly of polynomial size w.r.t. the size of the $\CP$-instance, it follows that both $w$ and $\alpha$ are of polynomial size $\mathcal{O}(kM^2)$. Therefore, the instance of $\minMisMatch_{\oneRepPat}$ can be computed in polynomial time, and our entire reduction is done in polynomial time. Moreover, we have shown that the instance $(w,\alpha,\Delta+M)$ of $\minMisMatch_{\oneRepPat}$ is answered positively if and only if the original instance of $\CP$ is answered positively. 

Finally, as the number of $x$ blocks in $\alpha$ is $k+1$, where $k$ is the number of input strings in the instance of $\CP$, and $\CP$ is $W[1]$-hard with respect to this parameter, it follows that $\minMisMatch_{\oneRepPat}$ is also $W[1]$-hard when the number of $k$-blocks in $\alpha$ is considered as parameter. This completes our proof.
\end{proof}

It is worth noting that the pattern $\alpha$ constructed in the reduction above is $k-1$-local (and not $k$-local): a witness marking sequence is $z_1<y_2<z_2<y_3<\ldots<z_{k-1}<y_k<x<y_1<z_k$. Thus, $\misMatch_\oneRepPat$ is W[1]-hard w.r.t. locality of the input pattern as well. Also, it is easy to see that $\scd(\alpha)=2$, and, by the results of \cite{ReidenbachS14}, this shows that the treewidth of the pattern $\alpha$, as defined in the same paper, is at most $3$. Thus, even for classes of patterns with constant $\scd$, number or repeated variables, or treewidth, the problems $\misMatch_P$ and $\minMisMatch_P$ can become intractable.

In Theorem \ref{thm:1repPTAS} we have shown that $\minMisMatch_{\oneRepPat}$ admits a polynomial time approximation scheme (for short, PTAS). We will show in the following that it does not admit an efficient PTAS (for short, EPTAS), unless $FPT=W[1]$. This means that there is no PTAS for $\minMisMatch_{\oneRepPat}$ such that the exponent of the polynomial in its running time is independent of the approximation ratio.

To show this, we consider an optimisation variant of the problem $\CP$, denoted $\mathtt{minCP}$. 
In this problem, for $k$ strings $w_1, \dots, w_k \in \sig{\ell}$ of length $\ell$ and an integer $m \in \mathbb{N}$ with $m\leq\ell$, we are interested in the smallest non-negative integer $\Delta$ for which there exist strings $s$, of length $m$, and $s_1, \ldots, s_k$, factors of length $m$ of each $w_1, \ldots, w_k$, respectively, such that $\sum_{i=1}^k \hdist{s_i}{s}=\Delta$.
In \cite{BoucherEPTAS}, it is shown that $\mathtt{minCP}$ has no EPTAS unless $FPT=W[1]$. We can use this result and the reduction from the Theorem \ref{thm:oneRepW1} to show the following result. 
\begin{restatable}{theorem}{thmoneRepNoEPTAS}\label{thm:oneRepNoEPTAS}
$\minMisMatch_\oneRepPat$ has no EPTAS unless $FPT=W[1]$.
\end{restatable}

\begin{proof}
Assume, for the sake of a contradiction, that $\minMisMatch_\oneRepPat$ has an EPTAS. That is, for an input word $w$ and an $\oneRepPat$-pattern $\alpha$, there exists a polynomial time algorithm which returns as answer to $\minMisMatch_\oneRepPat$ a value $\delta'\leq (1+\epsilon) \hdist{\alpha}{w}$, and the exponent of the polynomial in its running time is independent of $\epsilon$. 

An algorithm for $\mathtt{minCP}$ would first implement the reduction in Theorem \ref{thm:oneRepW1} to obtain a word $w$ and a pattern $\alpha$. Then it uses the EPTAS for $\minMisMatch_{\oneRepPat}$ to approximate the distance between $\alpha$ and $w$ with approximation ratio $(1+\frac{\epsilon}{2m})$. Assuming that this EPTAS returns the value $D$, the answer returned by this algorithm for the $\mathtt{minCP}$ problem is $D-m$. 

As explained in the proof of Theorem \ref{thm:oneRepW1}, it is easy to see that the distance between the word $w$ and the pattern $\alpha$ constructed in the respective reduction is $m+\Delta$, if $\Delta$ is the answer to the instance of the $\mathtt{minCP}$ problem. Thus, the value $D $ returned by the EPTAS for $\minMisMatch_\oneRepPat$ fulfils $m+\Delta\leq D\leq (1+\frac{\epsilon}{2m})(m+\Delta)$. So, we have $\Delta \leq D-m\leq \frac{\epsilon}{2}+(1+\frac{\epsilon}{2m})\Delta$. We get that $\Delta \leq D-m\leq (1+\frac{\epsilon}{2m}+\frac{\epsilon}{2\Delta})\Delta\leq (1+\epsilon)\Delta.$ So, indeed, $D-m$ would be a $(1+\epsilon)-$approximation of $\Delta$. 

Therefore, this would yield an EPTAS for $\mathtt{minCP}$. This is a contradiction to the results reported in \cite{BoucherEPTAS}, where it was shown that such an EPTAS does not exist, unless $FPT=W[1]$. This concludes our proof. 
\end{proof}

%

\newpage
\bibliography{main.bbl}

\end{document}